\documentclass{article}
\usepackage[utf8]{inputenc}

\usepackage{amsmath,amssymb,amsthm,fullpage,mathrsfs,pgf,tikz,caption,subcaption,mathtools,cleveref,bm,bbm,blkarray,array,accents}

\usepackage[numbers]{natbib}

\usepackage[margin=1in]{geometry}

\usepackage{enumitem}
\usepackage{pgfplots}

\newcommand{\CW}{\textup{CW}}
\newcommand{\ubar}[1]{\underaccent{\bar}{#1}}
\newcommand{\irr}{\textup{Irr}}

\usepackage[ruled,vlined,linesnumbered,noresetcount]{algorithm2e}
\let\oldnl\nl
\newcommand{\nonl}{\renewcommand{\nl}{\let\nl\oldnl}}

\usepackage[parfill]{parskip}
\newtheorem{theorem}{Theorem}[section]
\newtheorem*{theorem*}{Theorem}
\newtheorem{corollary}[theorem]{Corollary}
\newtheorem{proposition}[theorem]{Proposition}
\newtheorem{lemma}[theorem]{Lemma}
\newtheorem{definition}[theorem]{Definition}
\newtheorem*{definition*}{Definition}

\newtheorem{example}[theorem]{Example}
\newtheorem{remark}[theorem]{Remark}

\DeclareMathOperator{\poly}{poly}
\DeclareMathOperator{\polylog}{polylog}

\title{Improving the Leading Constant of Matrix Multiplication}

\author{Josh Alman\thanks{Department of Computer Science, Columbia University, NY 10027. Email: \texttt{josh@cs.columbia.edu}.} \and Hantao Yu\thanks{Department of Computer Science,
Columbia University, NY 10027. Email: \texttt{hantao.yu@columbia.edu}.}}

\begin{document}

\maketitle

\begin{abstract}
    Algebraic matrix multiplication algorithms are designed by bounding the rank of matrix multiplication tensors, and then using a recursive method. However, designing algorithms in this way quickly leads to large constant factors: if one proves that the tensor for multiplying $n \times n$ matrices has rank $\leq t$, then the resulting recurrence shows that $M \times M$ matrices can be multiplied using $O(n^2 \cdot M^{\log_n t})$ operations, where the leading constant scales proportionally to $n^2$. Even modest increases in $n$ can blow up the leading constant too much to be worth the slight decrease in the exponent of $M$. Meanwhile, the asymptotically best algorithms use very large $n$, such that $n^2$ is larger than the number of atoms in the visible universe!

    In this paper, we give new ways to use tensor rank bounds to design matrix multiplication algorithms, which lead to smaller leading constants than the standard recursive method.  
     Our main result shows that, if the tensor for multiplying $n \times n$ matrices has rank $\leq t$, then $M \times M$ matrices can be multiplied using only $n^{O(1/(\log n)^{0.33})} \cdot M^{\log_n t}$ operations. In other words, we improve the leading constant in general from $O(n^2)$ to $n^{O(1/(\log n)^{0.33})} < n^{o(1)}$. 

     We then apply this and further improve the leading constant in a number of situations of interest. We show that, in the popularly-conjectured case where $\omega=2$, a new, different recursive approach can lead to an improvement. We also show that the leading constant of the current asymptotically fastest matrix multiplication algorithm, and any algorithm designed using the group-theoretic method, can be further improved by taking advantage of additional structure of the underlying tensor identities.

     Our algorithms use new ways to manipulate linear transforms defined by Kronecker powers of matrices, as well as a new algorithm for very rectangular matrix multiplication. In many cases, we entirely avoid applying a large tensor rank identity by instead manipulating it implicitly or decomposing it into smaller, more manageable tensors. 
\end{abstract}

\section{Introduction}

In 1969, Strassen~\cite{Strassen} showed that it is possible to multiply two $M\times M$ matrices in faster than cubic time. His algorithm is based on an algebraic identity which shows that the \emph{rank} of the tensor for multiplying $2 \times 2$ matrices is at most $7$. (Equivalently, he showed how to multiply $2 \times 2$ matrices with non-commuting entries using only $7$ multiplications.) Strassen applied this algorithm in a recursive way, showing that the number $T(M)$ of operations to multiply two $M \times M$ matrices satisfies $T(M) \leq 7 T(M/2) + O(M^2)$, which solves to $T(M) \leq O(M^{\log_2(7)}) \leq O(M^{2.81})$.

Since then, a long line of work has gone into using this recursive idea to design even faster matrix multiplication algorithms. It is known that the tensor for multiplying $2 \times 2$ matrices does not have rank less than $7$~\cite{WINOGRAD1971381,hopcroft1971minimizing}, and so later work focused on bounding the ranks of larger matrix multiplication tensors. Notably, in later work, Strassen~\cite{strassen1973vermeidung} showed that any algebraic algorithm for multiplying matrices can be converted into a tensor rank bound which achieves the same exponent, and so algorithm designers have focused on this approach without loss of generality.

We write $\langle a,b,c \rangle$ to denote the $a \times b \times c$ matrix multiplication tensor (i.e., the computational problem of multiplying an $a \times b$ matrix and a $b \times c$ matrix). Given a matrix multiplication tensor $\langle n,n,n\rangle$ with rank at most $t$, one can apply Strassen's recursive idea again to see that the number $T(M)$ of operations to multiply two $M \times M$ matrices satisfies $T(M) \leq t \cdot T(M/n) + O(t \cdot M^2)$, which solves to $T(M) \leq O(n^2 \cdot M^{\log_n(t)})$. Normally one thinks of $n$ and $t$ as constants and writes the solution to this recurrence as $T(M) \leq O(M^{\log_nt})$. However, to achieve a good exponent $\log_n t$, current matrix multiplication algorithms require $n$ to be extremely large, which means that the leading constant $n^2$ is frequently very substantial. This paper focuses on decreasing this dependence on $n$.

Roughly speaking, there have been two regimes of work on bounding the rank of $\langle n,n,n \rangle$: When $n$ is small, to design more practical algorithms, and when $n$ is large, to get the asymptotically best algorithms. However, even in the small $n$ regime, while Strassen's algorithm has been implemented and used in practice (see, e.g.,~\cite{huang2016strassen} and its introduction), the rank bounds for larger $n$ still do not seem to be effective.

The fact that using an identity for $\langle n,n,n \rangle$ results in a leading constant which scales proportionally to $n^2$ can give a heuristic explanation for this. The smallest known tensor which achieves an improved exponent over Strassen's algorithm is of the form $\langle 12, 12, 12\rangle$ with rank at most $1040$ ~\cite{hopcroft1971minimizing,smirnov2013bilinear} (see also~\cite{sedoglavic2019yet}). Going from $n=2$ to $n=12$ improves the exponent by $\log_2(7) - \log_{12}(1040) \approx 0.01$, but increases the leading constant by about $(12/2)^2 = 36$. Hence, for the algorithm resulting from the $n=12$ tensor to be faster for multiplying $M \times M$ matrices, we need roughly $M^{0.01} > 36$, i.e., $M > 4 \times 10^{155}$. This is much larger than realistic inputs; in fact, it's approximately the square of the number of atoms in the visible universe! (We emphasize that this is just a heuristic argument; the exact leading constant of the running time depends on other details of the tensor rank expression like its sparsity.)

Our bounds on the rank of $\langle n,n,n \rangle$ are not known to be optimal for any $n>2$.\footnote{Even the rank of $\langle 3,3,3 \rangle$ is only known to be between $19$ and $23$~\cite{blaser2003complexity}. If it were in fact $19$, it would lead to an exponent of $\log_3(19) < 2.681$, surpassing Strassen's algorithm.} More generally, it is only known that the rank of $\langle n,n,n \rangle$ is at least $3n^2 - o(n^2)$~\cite{landsberg2014new,massarenti2013rank}. Nonetheless, there have been many attempts to improve these bounds using numerical methods, and it is generally believed that our rank bounds for small $n$ cannot be improved much.

Meanwhile, to achieve the best possible exponent of matrix multiplication, researchers focus on the regime where $n$ is very large. The exponent of matrix multiplication, $\omega$, is defined as a limit $$\omega:= \liminf_{n \to \infty} \log_n(R(\langle n,n,n \rangle))$$ where $R(\langle n,n,n \rangle)$ denotes the rank of the $\langle n,n,n \rangle$ tensor. Many intricate algorithmic techniques have been developed which yield improved upper bounds on $\omega$, but which require taking large $n$ so that the asymptotic behaviors of extremal combinatorial constructions kick in. Thus, these techniques seem to necessarily lead to impractical algorithms.

In fact, it's known that working with increasingly large $n$ is necessary to get the best exponent. Indeed, Coppersmith and Winograd~\cite{coppersmith1982asymptotic} gave a surprising procedure which takes in as input any tensor $\langle n,n,n \rangle$ with rank at most $t$, and outputs a new tensor $\langle n',n',n' \rangle$ with rank at most $t'$ where $n'$ is much larger than $n$, but $\log_{n'}(t') < \log_n(t)$. In other words, it coverts any input tensor rank identity into a new one which achieves a better exponent, but substantially increases $n$. Hence, it's necessary to define $\omega$ as a limit, since the exponent gets better as $n$ gets larger.

The current best bound on $\omega$, which achieves $\omega < 2.372$, uses an approach by Coppersmith and Winograd~\cite{CW87} which has been further refined by a recent line of work~\cite{davie2013improved,williams2012multiplying,le2014powers,AW21,DWZ22,williams2023new,ADWXXZ24}. Rather than give a single tensor, this line of work also gives a sequence of identities which get better as $n$ gets larger, and they only give a calculation of the exponent in the limit. This algorithm is typically described as impractical, and we quantify this assertion here: In Appendix~\ref{sec:backofenvelope} below, we perform an optimistic (underestimate) calculation of how large $n$ needs to be to get close to this exponent using the Coppersmith-Winograd approach. We find that, even to get a better exponent than $2.5$, $n$ needs to be \emph{substantially} larger than the number of atoms in the visible universe. Moreover, we will see that many of the core algorithmic techniques used by the approach (including the asymptotic sum inequality, sets avoiding arithmetic progressions, border rank bounds, and multinomial coefficient asymptotics) each seemingly necessarily contribute to $n$ being so large.

In summary, in both the small $n$ and large $n$ regimes, a leading constant of $n^2$ is prohibitively large compared to how much we're able to improve the rank (and hence exponent) compared to Strassen's algorithm. In this paper, we investigate whether the same tensor $\langle n,n,n\rangle$ can be used to design improved algorithms with a smaller leading coefficient or low-order terms.

In fact, there is some prior work which slightly improves on the $n^2$ leading coefficient. The classic technique of Lupanov~\cite{lupanov1956rectifier} can slightly decrease the operation count needed to compute any linear transformation, including the `encoding' and `decoding' transformations of a bilinear algorithm. It can thus be used to achieve a leading constant of $O(n^2 / \log n)$. A recent line of work~\cite{CH17,KS20,BS19,hadas2023towards} improved the constants one gets from particular small identities like Strassen's $\langle 2,2,2\rangle$ tensor (with rank $7$) and the $\langle 3,3,3\rangle$ tensor (with rank at most $23$) using a `sparse decomposition' technique; one could also generalize this to larger $n$, although doing so (for a particular identity) is not straightforward. However, simple counting arguments show that neither of these techniques can improve the leading constant in general beyond $O(n^2 / \log n)$.\footnote{We note that the sparse decomposition technique can work better than this for specific tensor rank identities that have additional structure to them. (Using additional structure of tensors to improve leading constants is a technique which is extensively explored in the literature, and which we will also employ below.) For one example, recent work by Hadas and Schwartz~\cite{hadas2023towards} showed that the tensor rank identity of Pan~\cite{pan78strassen} for $\langle 70, 70, 70 \rangle$ has a particular structure (its encoding and decoding matrices have duplicate rows and columns) which leads to a leading constant of only $2$, although in exchange, some lower-order terms of their operation count have larger constants. Later algorithms do not seem to have such a structure.} See \Cref{sec: counting barrier} below for a detailed explanation. Hence, an improvement of the leading constant beyond $O(n^2 / \log n)$ requires new techniques.

\section{Our Results}

In this paper, we give new ways to convert identities bounding the rank of $\langle n,n,n \rangle$ into matrix multiplication algorithms, improving the leading constant and low-order terms from Strassen's recursive approach. Our results focus on four regimes:
\begin{enumerate}
    \item An improved algorithm for general $n$ (\Cref{sec: better leading constants via rectangular MM});
    \item A further improvement to the current asymptotically fastest algorithm (\Cref{sec: application to current best algo});
    \item A further improvement to algorithms designed using the group-theoretic method (\Cref{sec: application to the group method});
    \item A further improvement assuming $\omega=2$ (\Cref{sec: application when omega is two}).
\end{enumerate}

We are optimistic that some of the techniques which have been developed to design asymptotically faster matrix multiplication algorithms can also be used in practical algorithm design by combining them with our new approach. For instance, in our improvement to the current asymptotically fastest algorithm (which is based on the Coppersmith-Winograd tensor), we show how the large tensor rank identity underlying the algorithm never needs to be explicitly constructed, so that the step size of our recursion depends on the smaller underlying Coppersmith-Winograd tensor rather than the large ultimate matrix multiplication tensor. (As we discuss more below, the Coppersmith-Winograd approach has other issues which prevent it from being practical, but we highlight some aspects of it whose contribution to the impracticality we are able to diminish.)

\subsection{Improvement for General $n$}

Our first main result gives a super-polynomial improvement to the leading constant, from $O(n^2)$ (from Strassen's recursive approach) and $O(n^2 / \log n)$ (from~\cite{lupanov1956rectifier}) to $2^{O((\log n)^{0.67})} < n^{o(1)}$.

\begin{theorem} \label{thm:main1}
    For any $\varepsilon > 0$, there is a function $c : \mathbb{N} \to \mathbb{R}_{>0}$ with $c(n) \leq n^{1/O((\log n)^{\frac{1}{3}-\varepsilon})} < n^{o(1)}$ such that: For any positive integers $n,t$ with $n>1$, and any field $\mathbb{F}$, given a matrix multiplication tensor $\langle n,n,n\rangle$ with rank at most $t$, we can design an algorithm for multiplying $M \times M$ matrices for all $M>n$ which uses $c(n) \cdot M^{\log_n(t)}$ arithmetic operations.
\end{theorem}

In other words, the same tensor $\langle n,n,n\rangle$ with rank at most $t$, which would give an algorithm running in time $O(n^2 M^{\log_n(t)})$ using Strassen's recursive approach, can be used to give an algorithm running in time $O(2^{(\log n)^{0.67}} M^{\log_n(t)})$ using our new approach. This circumvents the ``$n^2 / \log n$ counting argument barrier'' which applies to previous approaches.

Perhaps surprisingly, a key ingredient in our algorithm for Theorem~\ref{thm:main1} is a new algorithm for \emph{rectangular} matrix multiplication which may be of independent interest. Most prior work on matrix multiplication uses some kind of reduction to smaller matrix multiplications, but we're unaware of previous approaches which reduce to matrix multiplication of a vastly different shape.

We ultimately reduce most of the calculations needed by our algorithm to the task of multiplying an $M \times M$ matrix with an $M \times M^K$ matrix for an exponent $K \approx \sqrt{\log M}$ which is growing as a function of $M$. Rectangular matrix multiplication like this is well-studied in the case when $K$ is a constant using the Coppersmith-Winograd approach~\cite{HP98,LU18,gall2024faster,williams2023new}, but further intricacies arise in the case of super-constant $K$ which must be addressed.

\subsection{Improving the Current Best Algorithm}

Next, we study the current asymptotically best algorithm for matrix multiplication, which comes from a recent line of improvements to the Coppersmith-Winograd algorithm~\cite{CW87,williams2012multiplying,le2014powers,AW21,DWZ22,williams2023new,ADWXXZ24}. Let $\omega_0 < 2.372$ denote the best known upper bound on $\omega$, which is achieved by this algorithm. In other words, this means that for any $\varepsilon>0$, they give a corresponding algorithm with exponent $\omega_0 + \varepsilon$.

As previously discussed, this algorithm (and others following this line of work) has an enormous leading constant; more precisely, for some (very) large constant $C(\varepsilon) \approx 2^{2^{\Theta(1/\varepsilon^2)}}$ in terms of $\varepsilon$, its operation count for multiplying $M \times M$ matrices is \begin{align}
\label{eq:introoldbound}
C(\varepsilon)^{1+\sqrt{\frac{\log M}{2^{\Theta(1/\varepsilon^2)}}}} \cdot M^{\omega_0 + \varepsilon} + O(M^2).
\end{align} As we detail in Section~\ref{sec: application to current best algo} below, this algorithm uses Sch{\"o}nhage's asymptotic sum inequality~\cite{Sch} to minimize the exponent of $M$, and optimizing the parameters of that technique leads to an exponent of the large constant $C(\varepsilon)$ which grows like $\sqrt{\log M}$ as the input size $M$ grows.

We calculate in Section~\ref{sec:backofenvelope} below that even for the simpler~\cite{CW87} algorithm, $C(0.1)$ is much larger than the number of atoms in the visible universe. $C(\varepsilon)$ is even larger for more recent algorithms~\cite{DWZ22,williams2023new,ADWXXZ24} which have doubly-exponential, $C(\varepsilon) \approx 2^{2^{\Theta(1/\varepsilon^2)}}$ scaling in terms of $\varepsilon$, primarily because of a new procedure of removing ``holes''; see the proof of \Cref{thm:curromega} below for more details.

Applying our Theorem~\ref{thm:main1} directly improves the leading constant; the exponent of $C(\varepsilon)$ in \Cref{eq:introoldbound} improves from $\Theta{(\log M)^{1/2}}$ to about $\Theta{(\log M)^{1/3}}$. However, by taking advantage of the specific structure of the current best algorithm, we are able to further improve this and entirely remove the scaling in terms of $M$.

\begin{theorem} \label{thm:maincw}
Let $\omega_0<2.372$ denote the current best bound on $\omega$. Fix any $\varepsilon>0$ and let $C(\varepsilon)$ be as above. For all positive integers $M$, one can multiply $M \times M$ matrices over any field using at most
\begin{align*}
   C(\varepsilon)^{O(\varepsilon^2)} \cdot M^{\omega_0 + \varepsilon} + C(\varepsilon)^{O(1/\varepsilon^2)} \cdot M^2
\end{align*} field operations.
\end{theorem}

Our new algorithm removes the scaling with $M$ in the exponent of the leading constant (in \Cref{eq:introoldbound}), improving it to only $C(\varepsilon)^{O(\varepsilon^2)}$. This improves on the previous $1+\sqrt{\frac{\log M}{2^{\Theta(1/\varepsilon^2)}}}$ even when $M$ is a small constant, since $O(\varepsilon^2)$ becomes a smaller exponent than $1$. There is a larger coefficient $C(\varepsilon)^{O(1/\varepsilon^2)}$ of the lower-order term $M^2$, but it again becomes negligible compared to $C(\varepsilon)^{1+\sqrt{\frac{\log M}{2^{\Theta(1/\varepsilon^2)}}}}$ for a slightly larger constant $M$.

We don't claim that this is a practical algorithm, but it is substantially improved. Perhaps most interestingly, our proof highlights some of the techniques used in the Coppersmith-Winograd approach which previously caused a blowup in the leading constant, but which can be mostly mitigated using our new approach. For instance, the Coppersmith-Winograd approach critically makes use of the asymptotic sum inequality~\cite{Sch}, and of converting from border rank identities to rank identities for matrix multiplication~\cite{bini}. These two techniques both require taking a large Kronecker power of the final tensor in order to diminish their contribution to the exponent, but our new approach avoids this. Roughly speaking, previous approaches bounded the rank of $\langle n,n,n \rangle$ where $n$ is large enough that the asymptotic behaviors of these techniques kick in (and thus the input matrices must be much larger than $n$). With our new approach, it suffices to multiply input matrices of size only $n$, and we can use much smaller rank identities to do this.

On the other hand, one key component of the Coppersmith-Winograd approach which seems to unavoidably lead to a large constant factor is the use of arithmetic progression-avoiding sets. It is known that there are subsets of $\{1,2,3,\ldots,N\}$ avoiding three-term arithmetic progressions of size $N^{1 - o(1)}$~\cite{SS42,B46,elkin2010improved}. The Coppersmith-Winograd approach only starts to kick in when multiplying matrices exponentially larger than the $N$ for which this $o(1)$ in the exponent becomes insignificant. Unfortunately, recent work~\cite{kelley2023strong} has shown that the current best bound on this $o(1)$ factor cannot be substantially improved. Hence, although this technique is a fundamental component of the Coppersmith-Winograd approach, a truly practical algorithm would likely need to entirely avoid it. (This explains why $C(\varepsilon)$ must still appear in our final operation count in \Cref{thm:maincw}.)

We note that our proof of \Cref{thm:maincw} is quite general, and applies to any algorithm designed by starting with a fixed tensor $\mathcal{T}$ of low border rank, taking a large Kronecker power $\mathcal{T}^{\otimes N}$, and then zeroing out (or monomial degenerating) the result to a direct sum of (partial) matrix multiplication tensors. This general approach encapsulates most known ways to design matrix multiplication algorithms; it is called the `galactic method' in prior work~\cite{alman2021limits}. In particular, all the record-holding bounds on $\omega$ in the history of fast matrix multiplication algorithm design fit within this framework, so we are optimistic that this approach could lead to some of them being useful for practical algorithm design. Perhaps the most prominent approach for designing matrix multiplication algorithms which does not immediately fit into this framework is the Group-Theoretic Approach (since it uses the tensor of a fixed group rather than large Kronecker powers of that tensor), which we discuss next.

\subsection{Improving Group-Theoretic Algorithms}

The Group-Theoretic Method is an approach introduced by Cohn and Umans~\cite{CU03} for designing matrix multiplication algorithms. Rather than directly using an identity bounding the rank of matrix multiplication, the method makes use of an embedding of matrix multiplications into a finite group $G$ (using subsets of the group satisfying the so-called \emph{simultaneous triple product property}~\cite{CKSU05}). Improving the qualities of this embedding improves the exponent of the resulting algorithm.

We prove a result analogous to Theorem~\ref{thm:maincw} for \emph{any} algorithm designed in this way. The key idea is to use structural properties of the tensor corresponding to the group $G$ in place of properties of the Coppersmith-Winograd tensor. Similar to above, since this method uses the asymptotic sum inequality, it naturally leads to a leading coefficient for multiplying $M \times M$ matrices that scales with the input size as $|G|^{\Theta(\sqrt{\log M})}$. We improve this to a constant independent of $M$:

\begin{theorem} (Informal) \label{thm:introgroup}
The leading constant for any algorithm designed using the Group-Theoretic Method for a finite group $G$ is at most $16 \cdot |G|^{1.5}$.
\end{theorem}

In fact, depending on the details of the group $G$ (and its representation theory) and of the embedding of matrix multiplications into $G$, we can often achieve an even smaller bound, improving the exponent of $1.5$ down to $1$; see \Cref{sec: application to the group method} for the formal statement.

We remark that the Group-Theoretic Method is able to simulate the Coppersmith-Winograd Method (by embedding the Coppersmith-Winograd tensor within a group tensor~\cite{CKSU05}), although this requires taking a large Kronecker power of both the group tensor and Coppersmith-Winograd tensor, and so using it in conjunction with our Theorem~\ref{thm:introgroup} does not come close to our Theorem~\ref{thm:maincw}.

\subsection{Further Improvement if $\omega = 2$}

It is frequently conjectured that $\omega = 2$, i.e., that there is a function $f : \mathbb{N} \to \mathbb{R}_{>0}$ with $\lim_{n \to \infty} f(n) = 0$ such that for all $n \in \mathbb{N}$, the tensor $\langle n,n,n\rangle$ has rank $n^{2 + f(n)} $. The only known lower bound is the aforementioned $n^{2 + f(n)} \geq 3n^2 - o(n^2)$ for all $n$, i.e., that $f(n) \geq \Omega(1 / \log n)$. Our third main result addresses the question: What running time do we actually get for matrix multiplication if $\omega = 2$?

It is particularly interesting to ask what $f(n)$ would be needed to achieve an operation count of $M^2 \polylog(M)$ for multiplying $M \times M$ matrices. Aside from being a natural question in its own right, the need for such efficient matrix multiplication algorithm has recently been highlighted by a key algorithmic technique in fine-grained complexity theory. Many algorithms aim to get slight improvements on a straightforward running time by using a careful reduction to matrix multiplication, where any super-logarithmic factors in the running time would swamp the other savings from the approach. These algorithms currently need to reduce to a rectangular matrix multiplication $(M \times M^{0.1} \times M)$ where such a running time is possible~\cite{coppersmith1997rectangular,williams2014faster}. This phenomenon arises in a number of results in these areas, including the fastest known algorithms for problems like Orthogonal Vectors~\cite{abboud2014more,chan2016deterministic}, All-Pairs Shortest Paths~\cite{williams2014faster,chan2016deterministic}, Closest Pair problems~\cite{alman2015probabilistic,alman2016polynomial,alman2020faster} and Fourier transforms~\cite{Josh21}, as well as in the algorithms used to prove state-of-the-art circuit lower bounds in complexity theory (e.g.,~\cite{williams2014nonuniform,williams2010improving} and the many results that have followed).

However, using Strassen's recursive approach, it's impossible to achieve running time $M^2 \polylog(M)$ for $M \times M$ matrix multiplication, no matter how small $f(n)$ is. Indeed, suppose we assume that the current best lower bound is essentially tight, and that $f(n) \leq a / \log n$ for some constant $a$ and all $n$. Using Strassen's approach, we would get that for any positive integers $n$ and $M$, the number of operations needed to multiply two $M \times M$ matrices, is $O(n^2 M^{2 + a / \log n})$. This is minimized by picking $n = 2^{O(\sqrt{\log M})}$, yielding a final operation count of $M^2 \cdot 2^{O(\sqrt{\log M})}$. Even using our new Theorem~\ref{thm:main1} improves this to only about $M^2 \cdot 2^{O(\log M)^{0.4}}$. In other words, the current techniques cannot achieve an operation count of $M^2 \polylog(M)$, no matter how small $f(n)$ is.

Our last result shows how to further improve Theorem~\ref{thm:main1} in this setting. Since the function $f$ may have strange patterns of growth which could be hard to control or bound, our most general result statement is in terms of a sum $h_f$  involving $f$ (which we bound  important cases after the theorem statement):
\begin{theorem} \label{thm:mainomega2}
For any function $f : \mathbb{N} \to \mathbb{R}_{>0}$, define 
\[
h_f(m) := \sum_{\ell=0}^{\log \log m} \frac{3}{2^{\ell+1}} \cdot f(m^{1/2^{\ell+2}}).
\] If a tensor $\langle n,n,n\rangle$ with rank at most $n^{2 + f(n)} $ exists for all $n$, then $M \times M$ matrices can be multiplied using at most $O(M^{2 + h_f(N)})$ arithmetic operations.
\end{theorem}

One can calculate that for most ``nicely-behaved'' functions $f$, we get $h_f(m) \leq O(f(m) \cdot \log\log m)$. To give two important examples:

\begin{corollary} \label{cor:main2}
    Suppose there is a constant $a$ such that, for all sufficiently large $n$, the tensor $\langle n,n,n\rangle$ has rank at most $an^2$. Then, $M \times M$ matrices can be multiplied using $O(M^2 (\log M)^b)$ operations for some constant $b$ that depends only on $a$. 
\end{corollary}

By our \Cref{cor:main2}, if $\langle n,n,n \rangle$ has rank $O(n^2)$, then we achieve the desired operation count of $M^2 \polylog M$ for $M \times M$ matrix multiplication, improving on $M^2 \cdot 2^{O((\log m)^{0.4})}$ which one gets from~\Cref{thm:main1}.

For another example:

\begin{corollary} \label{cor:main3}
    Suppose there are constants $a,b$ such that, for all sufficiently large $n$, the tensor $\langle n,n,n \rangle$ has rank at most $a \cdot n^2 \cdot (\log n)^b$. Then, $N \times N$ matrices can be multiplied using $N^2 \cdot 2^{O((\log \log N)^2)}$ operations.
\end{corollary}

We prove Theorem~\ref{thm:mainomega2} by using a \emph{novel} recursive approach to matrix multiplication. 
Similar to Theorem~\ref{thm:main1}, a key idea in the proof of Theorem~\ref{thm:mainomega2} is a reduction from square to rectangular matrix multiplication. However, using the assumption that $\omega = 2$, we are able to perform that rectangular matrix multiplication even more quickly than using the Coppersmith-Winograd approach. In particular, our algorithm recursively performs that rectangular matrix multiplication, leading to a complicated recurrence to determine the final operation count and thus the sum $h_f$ above.

\section{Technique Overview}

\subsection{Encoding and Decoding matrices}
\label{sec: encoding_and_decoding_matrices}

We begin by recalling basic facts about bilinear algorithms. 
Each tensor $\langle n,m,d \rangle$ over a field $\mathbb{F}$ with rank at most $t$ is associated with two encoding matrices $X \in \mathbb{F}^{t \times (n\cdot m)}, Y \in \mathbb{F}^{t \times (m\cdot d)}$ and a decoding matrix $Z \in \mathbb{F}^{(n\cdot d) \times t}$. The operations of a Strassen-like algorithm to multiply $A \in \mathbb{F}^{n \times m}$ and $B \in \mathbb{F}^{m \times d}$ can be arranged to work as follows:
\begin{enumerate}
    \item Multiply $X$ times a vector of length $nm$ consisting of the entries of the matrix $A$, and multiply $Y$ times a vector of length $md$ consisting of the entries of the matrix $B$. The results are both vectors of length $t$.
    \item Compute the entry-wise multiplication of those two vectors of length $t$.
    \item Compute $Z^T$ times that vector of length $t$ to get a vector of length $nd$, which consists of the entries of the product $AB$.
\end{enumerate}
Therefore, the total number of arithmetic operations consists of two parts:
\begin{enumerate}
    \item The number of operations needed for multiplying the matrices $X,Y,Z$ times input vectors, and 
    \item The number of operations needed for the entry-wise vector multiplication.
\end{enumerate}
Notice that the number of operations needed for the entry-wise vector multiplication will simply be $t$ because we are entry-wise multiplying two vectors of length $t$.

Let $\otimes$ denote the Kronecker product (this and other notions from the literature on matrix multiplication algorithms are defined in \Cref{sec:prelims} below for the unfamiliar reader). 
For any positive integer $k$, the matrices $X^{\otimes k},Y^{\otimes k},Z^{\otimes k}$ will become valid encoding and decoding matrices for the tensor $\langle n^k,m^k,d^k\rangle$ with rank at most $t^k$. Therefore, in order to multiply increasingly larger input matrices (as $k$ grows), we can simply use the same algorithm but replace $X,Y,Z$ by $X^{\otimes k},Y^{\otimes k},Z^{\otimes k}$ and count the number of arithmetic operations needed for these matrices in the algorithm.

Hence, our focus will be on designing fast algorithms for multiplying Kronecker powers of fixed matrices with input vectors. This type of task has been well-studied recently in the context of computing Walsh-Hadamard transforms~\cite{Josh21,alman2023smaller,sergeev2022notes,alman2022faster}, and a number of our techniques are inspired by that line of work.

\subsection{Reduction to Rectangular Matrix Multiplication}

The standard approach for multiplying the Kronecker power $X^{\otimes k}$ of a matrix $X$ times a vector uses recursion (or the Mixed-Product property of the Kronecker product) to reduce to performing many multiplications of $X$ times a vector. Performing the multiplication in this way recovers Strassen's recursive approach with leading constant $O(n^2)$ from $\langle n,n,n\rangle$.

However, a key remark is that many of these multiplications by $X$ could be done in parallel, which exactly corresponds to another (smaller) matrix multiplication. For instance, the first step of this algorithm breaks the entries of input matrix $A \in \mathbb{F}^{n^k \times m^k}$ into $(nm)^{k-1}$ vectors of length $nm$, and multiplies $X$ by each of them. This exactly corresponds to performing a $t \times nm \times (nm)^{k-1}$ matrix multiplication. As $k$ scales with the input size, but we think of $t,n,m$ as (large) constants, this is a \emph{very} rectangular matrix multiplication, where the third dimension is substantially larger than the first two.

\subsection{First Attempt using Rectangular Matrix Multiplication}

There has been much prior work on rectangular matrix multiplication~\cite{CW87,HP98,LU18,gall2024faster,williams2023new}, which generally focuses on $M \times M \times M^K$ matrix multiplication for constant $K$. However, in our case, we would actually like to multiply such matrices where $K$ is a growing function of $M$ (which scales logarithmically with the input size to our matrix multiplication). 
Any algorithm for a fixed $K$ could be used, but we could hope to design an even faster algorithm than the prior algorithms in this previously-unexplored super-constant $K$ regime. We modify these algorithms from prior work to design a new algorithm for this problem.

As a first attempt, we use the Coppersmith-Winograd approach as in~\cite{HP98}. It is not hard to see (and has been previously observed~\cite{Josh21}) that this approach can give an algorithm using $M^{K+1+O(1/\log K)}$ operations for any constant $K$. (There is a natural lower bound of $M^{K+1}$, which is the output size.) Although prior instantiations of this approach like \cite{HP98} assume $K$ is a constant and use this fact to simplify some steps of the algorithm, one can perform a careful analysis of the hidden terms, and find that $M \times M \times M^K$ matrix multiplication can indeed be performed in $M^{K+1+O(1/ \log K)}$ operations as long as $K \leq O(\sqrt{\log M})$ is not too big. (For larger $K$, some hidden terms become significant.)

This suffices to combine with our approach, and use $\langle n,n,n \rangle$ to design an algorithm for $M \times M \times M$ matrix multiplication using $n^{O(1/\log \log n)} \cdot M^{\log_n(t)}$ operations. In other words, it achieves leading constant $n^{O(1/\log \log n)}$, which is already of the form $n^{o(1)}$, but not as good as the bound $n^{O(1/(\log n)^{0.33})}$ we are ultimately able to achieve.

Notably, the rectangular matrix multiplication algorithm here uses the tensor $\langle M,M,M^K \rangle$ with rank at most $M^{1 + K + O(1/\log K))}$ for very rectangular matrix multiplication which one gets from the Coppersmith-Winograd tensor, and then simply applies it in the usual recursive way. One could in principle use recursion to improve the leading constant of \emph{this} smaller rectangular matrix multiplication, although one can verify that this doesn't have a large impact; the $M^{O(1 / \log K)}$ factor is ultimately more significant than the straightforward leading constant.

\subsection{Improvement from properties of the Coppersmith-Winograd tensor}

We nonetheless further speed up the rectangular matrix multiplication algorithm to nearly remove the $O(1/\log K)$ term in the exponent of the running time. Our improvement is based on three main observations, which take advantage of the fact that the tensor $\langle m,m,m^K \rangle$ with rank at most $m^{1 + K + O(1/\log K)}$ was not just given to us as a black box, but that we know the details of how it comes from a zeroing out of a power of the Coppersmith-Winograd tensor.

\paragraph{Previous Parameters are Suboptimal for Large $K$.} The Coppersmith-Winograd tensor $\CW_q$ consists of three `corner terms' of the form $\langle 1,1,1 \rangle$, and three `edge terms' of the form $\langle q,1,1 \rangle$ or a rotation. When analyzing it, one assigns weights to each term, which determine the parameters $H, a$ such that $\CW_q^{\otimes N}$ can zero out into the direct sum of $H$ copies of $\langle a,a,a^K \rangle$. Roughly, the weights to the edge terms determine the shape of matrix multiplication one gets, so one of the three edge terms must be assigned $K$ times more weight than the other two. Meanwhile, the entropy of the weight distribution determines how large $H$ is.

When $K$ is a constant, the three edge terms are assigned significantly more weight than the three corner terms since they are more valuable; prior analysis~\cite{HP98} makes this assumption to simplify the analysis. However, when $K$ is super-constant, one of the three edge terms must be assigned so much weight that the three corner terms are assigned roughly the same weight as the other two edge terms. Breaking this assumption of prior work allows an improved exponent. 

\paragraph{Using the Encoding and Decoding Matrices of $\CW_q$.} Up to now, we have shown that Kronecker powers $\CW_q^{\otimes N}$ zero out into a direct sum of $H$ copies of $\langle a,a,a^k \rangle$. For notational simplicity, let $\mathcal{T}$ denote this direct sum of matrix multiplication tensors. Normally, by taking a Kronecker power of the border rank expression for $\CW_q$, we would also get a border rank upper bound for $\mathcal{T}$, and directly use this to get ``border'' encoding and decoding matrices for $\mathcal{T}$. We observe that the border encoding matrices for $\mathcal{T}$ are themselves (zeroing outs of) Kronecker powers of the border encoding matrices for $\CW_q$, which are much smaller. The leading constant of this algorithm thus comes from the leading constants of the border encoding matrices of $\CW_q$, which are just barely growing with $k$ (since we pick a super-constant $q$ in terms of $k$ in our analysis above).

One issue which arises with this idea is that, rather than using border encoding and decoding matrices corresponding to a border rank expression for $\mathcal{T}$, our matrix multiplication algorithm must use encoding and decoding matrices corresponding to a rank expression for $\mathcal{T}$. The typical way of converting border rank expressions to rank expressions of Bini~\cite{bini} picks out particular terms from the border encoding and decoding matrices, making them lose their Kronecker power structure. We instead convert from border rank to rank using polynomial interpolation (substituting values for the formal variable $\lambda$ of the border rank expression, and taking a linear combination of the resulting tensors) in a way which maintains the structure of the matrices as sums of Kronecker powers of fixed small matrices. In the case where we're working over a finite field (and there may not be enough different values of $\lambda$ for this interpolation), we instead first reduce to matrix multiplication over a large enough extension field, and prove that this has a negligible impact on the operation count since the size of the field we need scales only logarithmically with the input size.

\paragraph{Mostly Replacing the Asymptotic Sum Inequality.} Normally, when given a rank upper bound for the direct sum of $H$ copies of $\langle a,a,a^k \rangle$, one applies the asymptotic sum inequality to convert this into a rank upper bound for a single matrix multiplication tensor $\langle a', a', a'^k \rangle$ for some $a' \gg a$ which achieves the same exponent. However, because this increases $a'$, we would like to avoid it; this is what leads to the ``leading constant'' scaling with the input size as discussed above. Even Coppersmith and Winograd~\cite[Section~10]{CW87} notice this issue and suggest first steps toward resolving it by simplifying some steps of its proof.

We instead use a different approach of minimizing how much the asymptotic sum inequality must be used at all. When using a recursive algorithm for matrix multiplication, each step reduces from a smaller number of larger matrix multiplications to a larger number of smaller matrix multiplications. After a few steps, once the number of matrix multiplications is at least $H$, we can directly apply our identity for $H$ copies of $\langle a,a,a^k \rangle$. We thus apply only a constant number of steps of the asymptotic sum inequality (in order to achieve the correct exponent), but then our given rank bound for a direct sum of matrix multiplications for all remaining steps. In this way, our leading constant mainly depends on $a$ rather than $a'$. (This idea is hinted at, for instance in Filmus' notes~\cite[Section~5]{filmus2012matrix}, as intuition for \emph{proving} the asymptotic sum inequality, but we're unaware of any prior work that formally applies it in an algorithm like this.)

\subsection{Improving the Current Best Matrix Multiplication Algorithm}

The three observations of the previous subsection apply, not just to the rectangular matrix multiplication algorithm we designed, but more generally to any algorithm gotten by zeroing out (or monomial degenerating) powers of a fixed tensor $\mathcal{T}$ into a direct sum of matrix multiplication tensors. Since the current best bound on $\omega$ also comes from powers of the Coppersmith-Winograd tensor $\CW_5$, we apply these same techniques to yield \Cref{thm:maincw}. 

That said, some care must be taken: the current best algorithm finds a zeroing out of powers of $\CW_5$ to a direct sum of \emph{partial} matrix multiplication tensors, and then shows that many such tensors can be summed together to yield a full matrix multiplication tensor. We modify our approach to incorporate this, which actually slightly worsens the resulting leading constant compared to our rectangular matrix multiplication algorithm (which uses a version of the Coppersmith-Winograd approach without holes to yield a better leading constant).

\subsection{Improvement to the Group-Theoretic Approach}

The Group-Theoretic Method involves bounding the rank of matrix multiplication by finding matrix multiplication tensors embedded within the tensor $\mathcal{T}_G$ of a group algebra of a finite group $G$. The original work on this approach~\cite{CU03} noted that $\mathcal{T}_G$ has particularly elegant structure: after applying a carefully-chosen change of basis (the Discrete Fourier Transform of the group $G$), it is reduced to a number of smaller matrix multiplications (depending on the dimensions of the irreducible representations of $G$). \cite{CU03} focused on finding a single matrix multiplication tensor embedded in $\mathcal{T}_G$, and (implicitly) noted that in this case, this structure of $\mathcal{T}_G$ translates to giving well-structured encoding and decoding matrices for matrix multiplication which lead to small leading constants.

However, starting with the follow-up work~\cite{CKSU05}, most work in this area has focused on embedding \emph{multiple} matrix multiplication tensors within $\mathcal{T}_G$. (The current best algorithms using the Group-Theoretic Method are designed in this way.) In this case, the asymptotic sum inequality must once again be applied to design a matrix multiplication algorithm (as in \cite[Theorem 5.5]{CKSU05}). However, in addition to the other issues with the asymptotic sum inequality discussed above, it also does not respect the structure of the rank expression for the original tensor $\mathcal{T}_G$, and the aforementioned structure of $\mathcal{T}_G$ seemingly cannot be used effectively to improve the leading constant. 

Using the same idea above of minimizing the use of the asymptotic sum inequality, we are mostly able to get around this issue. Again, only a small number of recursive steps are applied using the asymptotic sum inequality, but most of the steps of our algorithm are able to directly use the rank expression of $\mathcal{T}_G$ itself and recover most of its benefits.

Then, by directly applying our improved generic matrix multiplication algorithm, we're able to achieve a further improvement. Since the Discrete Fourier Transform of the group $G$ reduces $\mathcal{T}_G$ to many smaller matrix multiplications, it actually gives a recursive algorithm where we then need an approach for performing these smaller multiplications. We apply our main result to improve the leading constant of these smaller operations and hence the leading constant of the entire algorithm.

\subsection{Improvement when $\omega=2$}

In the case when $\omega=2$, we directly use this, rather than the Coppersmith-Winograd approach, to design an even faster algorithm for the required rectangular matrix multiplication and prove \Cref{thm:mainomega2}. Recall that, given a tensor $\langle n,n,n \rangle$ with rank at most $t$, our goal is to bound the number of operations $T(n^k)$ needed to multiply $n^k \times n^k$ matrices for growing $k$. We have thusfar reduced this to very rectangular $n^2 \times n^2 \times n^{2(k-1)}$ matrix multiplication. However, by a simple blocking strategy, this in turn reduces to $n^{2(k-2)}$ copies of $n^2 \times n^2 \times n^2$ matrix multiplication. We thus can recursively bound $T(n^k)$ in terms of $T(n^2)$, the number of operations needed to perform $n^2 \times n^2 \times n^2$ matrix multiplication. In the case when $\omega=2$, so that $T(n^2)$ is quite close to $n^4$, this can be a very efficient approach, even faster than the Coppersmith-Winograd approach above.

\section{Preliminaries} \label{sec:prelims}

\subsection{Matrix Multiplication Tensors}

Throughout this paper we will be designing algebraic algorithms over a field $\mathbb{F}$. We will work in the arithmetic circuit model of computation where we count the number of field operations ($+,-,\times,\div$) of our algorithms. In principle this means we allow nonuniform algorithms, although all the algorithms given in this paper will be uniform. We may interchangeably refer to the ``running time'' and ``operation count'' of an algorithm.

The problem of matrix multiplication is: given as input $A \in \mathbb{F}^{n_1 \times n_2}$ and $B \in \mathbb{F}^{n_2 \times n_3}$, the goal is to compute the matrix $C \in \mathbb{F}^{n_1 \times n_3}$ given by
\[
C_{ij} = \sum_{k = 1}^{n_2}A_{ik}\cdot B_{kj}.
\] We use $T(n_1,n_2,n_3)$ to denote the number of arithmetic operations needed to compute this product.
Most of the algorithms and techniques in this paper work equally well over any field $\mathbb{F}$; when this is not the case, we will be clear about what field we are working over.

\begin{definition}[Tensor]
Let $\mathbb{F}$ be any field. Let $A = \{a_{1},\ldots,a_{|A|}\}, B = \{b_1,\ldots,b_{|B|}\}, C = \{c_1,\ldots,c_{|C|}\}$ be finite sets of variables. A {\em tensor} $\mathcal{T}$ over $A,B,C$ is a trilinear form
\[
\mathcal{T} = \sum_{i=1}^{|A|}\sum_{j=1}^{|B|}\sum_{k=1}^{|C|}\mathcal{T}[i,j,k]\cdot a_{i}b_{j}c_{k},
\] where $\mathcal{T}[i,j,k] \in \mathbb{F}$ for all $i,j,k$. We use $|\mathcal{T}|$ to denote the output size of a tensor; in this case, $|\mathcal{T}| = |C|$.
\end{definition} 

Given two tensors on the same set of variables $A \times B \times C$ as
\[
\mathcal{T} = \sum_{i=1}^{|A|}\sum_{j=1}^{|B|}\sum_{k=1}^{|C|}\mathcal{T}[i,j,k]\cdot a_{i}b_{j}c_{k},\ \ \ \mathcal{T}' = \sum_{i=1}^{|A|}\sum_{j=1}^{|B|}\sum_{k=1}^{|C|}\mathcal{T}'[i,j,k]\cdot a_{i}b_{j}c_{k},
\] we say their {\em sum} $\mathcal{T}+\mathcal{T}'$ is a tensor whose coefficient of $a_ib_jc_k$ is $\mathcal{T}[i,j,k]+\mathcal{T}'[i,j,k]$, where $(a_i,b_j,c_k) \in A \times B \times C$ and $\mathcal{T}[i,j,k],\mathcal{T}'[i,j,k]$ are the coefficients of $a_ib_jc_k$ in $\mathcal{T},\mathcal{T}'$ respectively. If 
\[
\mathcal{T} = \sum_{i=1}^{|A|}\sum_{j=1}^{|B|}\sum_{k=1}^{|C|}\mathcal{T}[i,j,k]\cdot a_{i}b_{j}c_{k},\ \ \ \mathcal{T}' = \sum_{i=1}^{|A'|}\sum_{j=1}^{|B'|}\sum_{k=1}^{|C'|}\mathcal{T}'[i,j,k]\cdot a_{i}'b_{j}'c_{k}'
\] are in different set of variables $A \times B \times C, A' \times B' \times C'$ such that $A\cap A' = B\cap B' = C\cap C' = \varnothing$, then we define the \emph{direct sum} of $\mathcal{T}$ and $\mathcal{T}'$ to be
\[
\mathcal{T}\oplus \mathcal{T}' = \sum_{i=1}^{|A|}\sum_{j=1}^{|B|}\sum_{k=1}^{|C|}\Big(\mathcal{T}[i,j,k]\cdot a_{i}b_{j}c_{k}+\mathcal{T}'[i,j,k]\cdot a_i'b_j'c_k'\Big)
\] over $(A\sqcup A')\times (B\sqcup B')\times (C\sqcup C')$. Given a positive integer $H$, we use $H\odot T$ to denote the direct sum of $H$ copies of $\mathcal{T}$.

\begin{definition}[Kronecker Product of Tensors]
The {\em Kronecker product} of $\mathcal{T}$ and $\mathcal{T}'$, denoted as $\mathcal{T} \otimes \mathcal{T}'$, is a tensor over $A \times A', B \times B', C \times C'$ given by
\[
\mathcal{T} \otimes \mathcal{T}' = \sum_{i=1}^{|A|}\sum_{j=1}^{|B|}\sum_{k=1}^{|C|}\sum_{i'=1}^{|A'|}\sum_{j'=1}^{|B'|}\sum_{k'=1}^{|C'|}\mathcal{T}[i,j,k]\cdot \mathcal{T}'[i',j',k']\cdot (a_i,a_{i'}')\cdot (b_j,b_{j'}')\cdot (c_k,c_{k'}').
\] For a positive integer $n$, we write $\mathcal{T}^{\otimes n} = \mathcal{T} \otimes \mathcal{T} \otimes \cdots \otimes \mathcal{T}$ ($n$ times) for the $n$-th Kronecker power of $\mathcal{T}$.
\end{definition}

\begin{definition}[Tensor Rank]
We say a tensor $\mathcal{T}$ has rank 1 if we can write
\[
\mathcal{T} = \left(\sum_{i=1}^{|A|}\alpha_i\cdot a_i\right)\left(\sum_{j=1}^{|B|}\beta_j\cdot b_j\right)\left(\sum_{k=1}^{|C|}\gamma_k\cdot c_k\right)
\] for some $\alpha_i,\beta_j,\gamma_k \in \mathbb{F}$. This is equivalent to saying $\mathcal{T}[i,j,k] = \alpha_i\cdot \beta_j\cdot \gamma_k$ for all $i,j,k$. The {\em rank} of a tensor $\mathcal{T}$, denoted as $R(\mathcal{T})$, is the minimum nonnegative integer such that there are rank 1 tensors $\mathcal{T}_1,\ldots,\mathcal{T}_{R(\mathcal{T})}$ that sum to $\mathcal{T}$. 
\end{definition}

Some basic properties of rank include
\begin{itemize}
    \item $R(\mathcal{T}+\mathcal{T}') \leq R(\mathcal{T})+R(\mathcal{T}')$.
    \item $R(\mathcal{T} \otimes \mathcal{T}') \leq R(\mathcal{T})\cdot R(\mathcal{T}')$.
\end{itemize}

We also define the notion of border rank, which is an approximation of tensor rank.

\begin{definition}[Border Rank]
We say a tensor $\mathcal{T}$ has border rank at most $r$, denoted by $\ubar{R}(\mathcal{T})$, if for some $d$ there are tensors $\mathcal{T}_h$ such that 
\[
\mathcal{T}+\sum_{h=1}^{3d}\varepsilon^hT_h = \sum_{\ell=1}^{r}\Big(\sum_{x \in X}\alpha_{x,\ell}(\varepsilon)\cdot x\Big)\Big(\sum_{y \in Y}\beta_{y,\ell}(\varepsilon)\cdot y\Big)\Big(\sum_{z \in Z}\gamma_{z,\ell}(\varepsilon)\cdot z\Big),
\] where $\alpha_{x,\ell}(\varepsilon),\beta_{y,\ell}(\varepsilon),\gamma_{z,\ell}(\varepsilon)$ are polynomials in $\varepsilon,\frac{1}{\varepsilon}$ of degree at most $d$.
\end{definition}

\begin{definition}[Matrix Multiplication Tensor] \label{def:mmtensor}
For positive integers $n,m,d$, the $n \times m \times d$ matrix multiplication tensor, denoted as $\langle n,m,d \rangle$, is a tensor over the variables
\[
A = \{a_{ij}: i\in [n], j\in [m]\}, B = \{b_{jk}:j \in [m], k \in [d]\}, C = \{c_{ki}: k\in [d], i\in [n]\},
\] given by 
\[
\langle n,m,d \rangle = \sum_{i \in [n]}\sum_{j \in [m]}\sum_{k \in [d]}a_{ij}b_{jk}c_{ki}.
\] 
\end{definition}

In Definition~\ref{def:mmtensor}, one can think of $A,B$ as two matrices that we want to multiply such that $a_{ij},b_{jk}$ are their entries. $c_{ki}$ denotes a variable for the $(k,i)$ entry of the output matrix, and its coefficient in the tensor $\langle n,m,d \rangle$ is the value of the $(k,i)$ entry in the product $A \times B$. Let $T(\langle n,m,d\rangle)$ denote the number of arithmetic operations needed to compute $\langle n,m,d\rangle$. The Baur-Strassen's theorem \cite{BS83} shows that $T(n,m,d)$ and $T(\langle n,m,d\rangle)$ are the same up to a constant factor.

\begin{definition}[Matrix Multiplication Exponent]
The matrix multiplication exponent, denoted as $\omega$, is defined to be
\[
\omega = \inf_{q \in \mathbb{N}}\log_q R(\langle q,q,q \rangle).
\]
\end{definition} 
It is often conjectured that $\omega = 2$, but the current best bounds based on the Coppersmith-Winograd method~\cite{CW87,davie2013improved,williams2012multiplying,le2014powers,AW21,DWZ22,williams2023new,ADWXXZ24} achieve $\omega < 2.372$.

Border rank allows us to upper bound $\omega$, which was proved by Sch\"{o}nhage \cite{Sch}. 

\begin{lemma}[Asymptotic Sum Inequality]
\label{lem: asymptotic sum inequality}
If $\underline{R}(\oplus_{i=1}^{\ell}\langle n_i,m_i,d_i\rangle) \leq r$, then $\sum_{i=1}^{\ell}(n_im_id_i)^{\omega/3} \leq r$.
\end{lemma}

A standard fact about matrix multiplication tensors we will use is that they compose well under Kronecker products:

\begin{lemma} 
Given two matrix multiplication tensors $\langle n,m,d \rangle$ and $\langle n',m',d' \rangle$, we have
\[
\langle n,m,d \rangle \otimes \langle n',m',d' \rangle = \langle n\cdot n', m\cdot m',d\cdot d' \rangle.
\]
\end{lemma}
\begin{proof}
Suppose the variables for $\langle n,m,d\rangle$ are
\[
A = \{a_{ij}: i\in [n],j \in [m]\}, B = \{b_{jk}: j\in [m], k \in [d]\}, C = \{c_{ki}:k\in [d],i\in [n]\}
\] and the variables for $\langle n',m',d'\rangle$ are 
\[
A' = \{a'_{ij}: i\in [n'],j \in [m']\}, B' = \{b'_{jk}: j\in [m'], k \in [d']\}, C = \{c'_{ki}:k\in [d'],i\in [n']\}.
\]
The variables in $\langle n,m,d \rangle \otimes \langle n',m',d' \rangle$ are over $A \times A', B\times B',C\times C'$. In particular, for any variable $(a_{ij},a'_{i'j'})\cdot (b_{k\ell},b'_{k'\ell'})\cdot (c_{pq},c'_{p'q'})$, its coefficient in $\langle n,m,d \rangle \otimes \langle n',m',d' \rangle$ is $1$ if and only if $j = k, \ell = p, q = i, j' = k', \ell' = p', q' = i'$ and $0$ otherwise. This is equivalent to saying $(i,i') = (q,q'), (j,j') = (k,k'),(p,p') = (\ell,\ell')$ which happens if and only if the coefficient of $(a_{ij},a'_{i'j'})\cdot (b_{k\ell},b'_{k'\ell'})\cdot (c_{pq},c'_{p'q'})$ in $\langle n\cdot n',m\cdot m',d\cdot d'\rangle$ is $1$ (and $0$ otherwise).
\end{proof}

We will also use the standard fact that we can multiply matrices with a ``divide-and-conquer" method.
\begin{lemma} \label{lem:extendMM}
For positive integers $n,m,d,n',m',d'$,
\begin{align*}
    T(n \cdot n',m,d) &\leq n' \cdot T(n,m,d), \\
    T(n ,m,d \cdot d') &\leq d' \cdot T(n,m,d), \\
    T(n ,m \cdot m',d ) &\leq m' \cdot T(n,m,d) + (m'-1)nd.   
\end{align*}
\end{lemma}

\begin{proof}
    We start with $T(n \cdot n',m,d) \leq n' \cdot T(n,m,d)$. Suppose we are multiplying $X \in \mathbb{F}^{nn' \times m}$ and $Y \in \mathbb{F}^{m \times d}$. We may partition the rows of $X$ into $n'$ matrices $X_1, \ldots, X_{n'} \in \mathbb{F}^{n \times m}$ and separately multiply each times $Y$. The desired result $X \times Y$ is a concatenation of the rows of those results.

    The proofs of the second and third inequalities are similar, except that for the third, the desired result will be the sum of the $m'$ individual results, which are each $n \times d$ matrices, so we need $(m'-1)nd$ operations to add them all.
\end{proof}

\subsection{Bilinear Algorithms for Matrix Multiplication} \label{sec: bilinear algorithms for matrix multiplication}

The idea of multiplying matrices using bilinear algorithms originated from \cite{Strassen}. Here we use the notations formalized in \cite{KKB88}.

\begin{definition}
A linear algorithm $A$ over $\mathbb{R}$ is a triple $(V,E,\lambda)$, where
\begin{enumerate}
    \item $(V,E)$ is a directed acyclic graph with vertex set $V$ and edge set $E$; and
    \item $\lambda: E \rightarrow \mathbb{R}-\{0\}$.
\end{enumerate} We denote by $V^I$ the subset of $V$ consisting of vertices of in-degree $0$ and $V^{O}$ the subset of $V$ consisting of vertices of out-degree $0$. We say that $|V^I|$ is the \textup{input size} of $A$ and $|V^O|$ is the \textup{output size} of $A$.
\end{definition}

\begin{definition}

Let $A = (V,E,\lambda)$ be a linear algorithm over $\mathbb{R}$ and $\beta:[|V^I|] \rightarrow V^I$ be any bijective labeling. With each vertex $v \in V$ we can associate an element $f(v) \in \mathbb{R}[X]$, called the linear form associated with $v$, as
\begin{align*}
f(v) =
\begin{cases}
x_{\beta^{-1}(v)} & v \in V^I\\[0.1in]
\displaystyle\sum_{(w,v)\in E}\lambda(w,u)\cdot f(w) & v \in V-V^I.
\end{cases}
\end{align*}
\end{definition}

\begin{definition}
We say that linear algorithm $A = (V,E,\lambda)$ computes a set of forms $\{f_1,\ldots,f_s\} \subset \mathbb{R}[X]$ if for each $f_i$ there exists a $v \in V^{O}$ such that $f_i = f(v)$. Suppose $|V^I| = t$ and $|V^O| = s$. We denote by $F_A$ the vector $(g_1,\ldots,g_s)^T$, where $g_j = f(\gamma(j))$. The matrix associated with algorithm $A$, denoted $M_A$, is the $s \times t$ matrix over $\mathbb{R}$ defined by $M_A\mathbf{x}_t = F_A$.
\end{definition}

\begin{figure}
    \centering
    \includegraphics[scale=0.2]{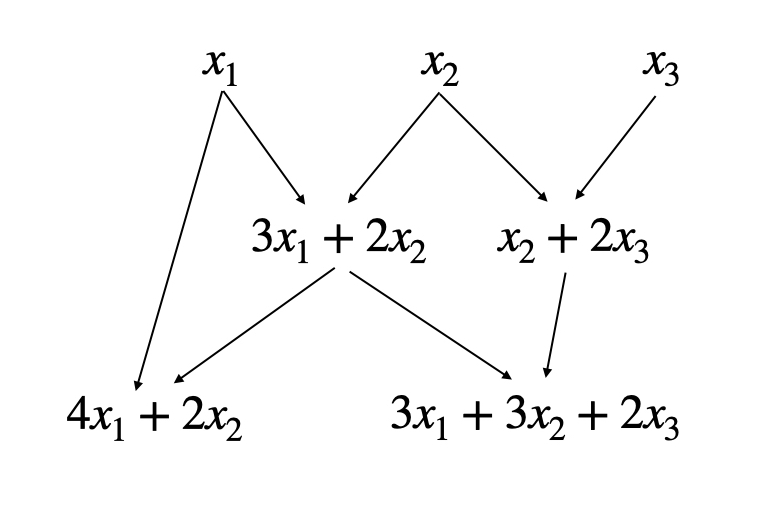}
    \caption{An example of a linear algorithm. The algorithm takes in $x_1,x_2,x_3$ as inputs and outputs $4x_1+2x_2,3x_1+3x_2+2x_3$. The matrix associated with it is 
    $\begin{pmatrix}
        4 & 2 & 0\\
        3 & 3 & 2
    \end{pmatrix}$.}
    \label{fig:enter-label}
\end{figure}

We use $T(A)$ to denote the minimal number of arithmetic operations needed for $A$. This is equivalent to the minimal number of arithmetic operations needed for multiplying $M_A$ with an arbitrary vector. We will also use $T(M)$ to denote the number of arithmetic operations needed for multiplying $A$ with a vector. As a result, $T(A) = T(M_A)$.

\begin{definition}
Let $A = (V_1,E_1,\lambda)$ be any linear algorithm over $\mathbb{R}$. We denote by $A^{\top}$ the linear algorithm over $\mathbb{R}$ obtained from $A$ from reversing all of the edges in $E$. More formally, $A^{\top} = (V_2,E_2,\lambda_2)$ where $V_2 = V_1, E_2 = \{(v,w):(w,v)\in E_1\}$ and $\lambda_2(v,w) = \lambda_1(w,v)$ for all $(w,v) \in E_1$. We refer to $A^T$ as the transpose of algorithm $A$.
\end{definition}

\cite{KKB88} proves that if we take the transpose of the matrix associated to a linear algorithm, then we get the matrix associated to the transpose of that linear algorithm.

\begin{theorem}[\cite{KKB88} Theorem 1]
\label{thm: matrix_transpose_algorithm}
Let $A$ be any linear algorithm over $\mathbb{R}$, then $M_{A^{\top}} = (M_A)^{\top}$.
\end{theorem}

\begin{lemma}[\cite{KKB88} Theorem 2]
\label{lem: matrix inverse time}
Suppose $A = (V,E,\lambda)$ is a linear algorithm such that $M_A$ does not have any zero row or column, then $T(A) = T(A^{\top})+|V^I|-|V^O|$.
\end{lemma}
\begin{proof}
For any linear algorithm $A$, we can reverse all the edges in $A$ to obtain $A^{\top}$. Since $M_A$ does not have any zero row or column, none of the input or output nodes in $V$ are isolated, and thus $T(A) = T(A^{\top})+|V^I|-|V^O|$.
\end{proof}

\begin{definition}[Bilinear algorithm]
A bilinear algorithm $B$ over $\mathbb{R}$ is a triple $(A_1,A_2,A_3)$ where
\begin{enumerate}
    \item $A_i = (V_i,E_i,\lambda_i)$ is a linear algorithm over $\mathbb{R}$ for $1 \leq i \leq 3$; and
    \item $|V_1^O| = |V_2^O| = |V_3^I|$. We call this value the \textup{rank} of $B$, denoted as $R(B)$.
\end{enumerate} $B$ is said to compute the set of bilinear forms
\[
(M_{A_3})^T[(M_{A_2}\mathbf{y}_s)\star (M_{A_{1}}\mathbf{x}_t)],
\] where $\star$ denotes component-wise product. We say that $M_{A_1},M_{A_2}$ are the \textup{encoding matrices} and $M_{A_3}$ is the \textup{decoding matrix} for $B$.
\end{definition}

We use $T(B)$ to denote the number of arithmetic operations needed for $B = (A_1,A_2,A_3)$, and we know from definition that 
\[
T(B) = T(A_1)+T(A_2)+T(A_3)+R(B).
\] We will use this formula throughout our paper when discussing the arithmetic complexity of bilinear algorithms for matrix multiplication. 

\begin{definition}[Bilinear algorithm for tensor]
Let $B = (A_1,A_2,A_3)$, $A_i = (V_i,E_i,\lambda_i)$ for $1 \leq i \leq 3$, be a bilinear algorithm. Let 
\[
\mathcal{T} = \sum_{i=1}^{n}\sum_{j=1}^{m}\sum_{k=1}^{d}\mathcal{T}[i,j,k]\cdot a_ib_jc_k
\] be any tensor. We say that $B$ is a bilinear algorithm for computing a tensor $\mathcal{T}$ if $B$ computes all coefficients of $c_k$, i.e.
\[
\Bigg\{\sum_{i=1}^{n}\sum_{j=1}^{m}\mathcal{T}[i,j,k]\cdot a_ib_j : \hspace{0.1cm}\forall k\Bigg\}.
\]
 
\end{definition}

A famous example of bilinear algorithms for tensor is Strassen's algorithm for computing $\langle 2,2,2\rangle$ with $7$ multiplications.

\begin{example}[Strassen's algorithm \cite{Strassen}]
Given two $2\times 2$ matrices $A = 
\begin{pmatrix}
A_{11} & A_{12}\\
A_{21} & A_{22}
\end{pmatrix},
B = 
\begin{pmatrix}
B_{11} & B_{12}\\
B_{21} & B_{22}
\end{pmatrix}$, we can multiply them using only 7 multiplications as follows. Compute
\begin{equation*}
    \begin{split}
        M_1 &= (A_{11}+A_{22})(B_{11}+B_{22})\\
        M_2 &= (A_{21}+A_{22})B_{11}\\
        M_3 &= A_{11}(B_{12}-B_{22})\\
        M_4 &= A_{22}(B_{21}-B_{11})\\
        M_5 &= (A_{11}+A_{12})B_{22}\\
        M_6 &= (A_{21}-A_{11})(B_{11}+B_{12})\\
        M_7 &= (A_{12}-A_{22})(B_{21}+B_{22})
    \end{split}
\end{equation*} and the product of $A,B$ can be expressed as
\[
A\cdot B =  
\begin{pmatrix}
A_{11}\cdot B_{11}+A_{12}\cdot B_{21} & A_{11}\cdot B_{12}+A_{12}\cdot B_{22}\\
A_{21}\cdot B_{11}+A_{22}\cdot B_{21} & A_{21}\cdot B_{12}+A_{22}\cdot B_{22}
\end{pmatrix} = 
\begin{pmatrix}
M_1+M_4-M_5+M_7 & M_3+M_5\\
M_2+M_4 & M_1-M_2+M_3+M_6
\end{pmatrix}.
\] If we write $B = (A_1,A_2,A_3)$ as Strassen's algorithm, the encoding and decoding matrices are 
\[
M_{A_1} = 
\begin{pmatrix}
1 & 0 & 0 & 1\\
0 & 0 & 1 & 1\\
1 & 0 & 0& 0\\
0 & 0 & 0 & 1\\
1 & 1 & 0 & 0\\
-1 & 0 & 1 & 0\\
0 & 1 & 0 & -1
\end{pmatrix},
M_{A_2} =
\begin{pmatrix}
1 & 0 & 0 & 1\\
1 & 0 & 0 & 0\\
0 & 1 & 0 & -1\\
-1 & 0 & 1 & 0\\
0 & 0 & 0 & 1\\
1 & 1 & 0 & 0\\
0 & 0 & 1 & 1
\end{pmatrix},
M_{A_3} = 
\begin{pmatrix}
1 & 0 & 0 & 1\\
0 & 0 & 1 & -1\\
0 & 1 & 0 & 1\\
1 & 0 & 1 & 0\\
-1 & 1 & 0 & 0\\
0 & 0 & 0 & 1\\
1 & 0 & 0 & 0
\end{pmatrix}.
\] 
\end{example}

A commonly used trick in upper bounding the arithmetic complexity of computing tensors is the ``zeroing out" technique.
\begin{definition}
We say that tensor $\mathcal{T}'$ can be \textup{zeroing out} from $\mathcal{T}$, denoted by $\mathcal
{T}' \leq \mathcal{T}$, if we can set some variables to $0$ in $\mathcal{T}$ such that it becomes $\mathcal{T}'$.
\end{definition}

One example is $\mathcal{T}' = x_0y_0z_0, \mathcal{T} = x_0y_0z_0+x_0y_0z_1$. We have $\mathcal{T}' \leq \mathcal{T}$ because we can set $z_1 = 0$ in $T$ such that it becomes $\mathcal{T}'$. As a result, suppose $\mathcal{T},\mathcal{T}'$ are both matrix multiplication tensors and $B = (A_1,A_2,A_3)$ is a bilinear algorithm for $\mathcal{T}$, then it can be used to compute $\mathcal{T}'$ with at most $T(B)$ arithmetic operations (since setting variables to zero will never increase the number of operations needed).

\subsection{Tensor Product of Bilinear Algorithms}

In general, if we have a $\langle n,m,d\rangle$ tensor with rank $t$, not only can we do $n \times m$ and $m \times d$ matrix multiplication as above, we can also do $n^k \times m^k$ and $m^k \times d^k$ matrix multiplication for any integer $k>1$. This is because if we have a tensor $\langle n,m,d\rangle$, then taking the $k$-th Kronecker power immediately gives us a new tensor $\langle n^k,m^k,d^k\rangle$ with rank at most $t^k$ as we will show next, so we can run the algorithm with the new tensor to multiply $n^k \times m^k$ and $m^k \times d^k$ matrices. This motivates the Kronecker product of matrices.

\begin{definition}[Kronecker Product of Matrices]
If $M$ is an $m \times n$ matrix and $N$ is a $p \times q$ matrix, the Kronecker product $A \otimes B$ is the $pm \times qn$ block matrix
\[
A \otimes B = 
\begin{pmatrix}
    a_{11}B & \cdots & a_{1n}B\\
    \vdots & \ddots & \vdots\\
    a_{m1}B & \cdots & a_{mn}B
\end{pmatrix}
\] For a positive integer $k$, we use $A^{\otimes k}$ to denote $A \otimes A \otimes \cdots A$ ($k$ times).
\end{definition}

The Kronecker product of matrices satisfies many elementary properties including bilinearity and associativity: given matrices $A,B,C$ of appropriate sizes and positive integer $k$,
\begin{enumerate}
    \item $A \otimes (B+C) = A\otimes B+A \otimes C$;
    \item $(A \otimes B) \otimes C = A \otimes (B \otimes C)$;
    \item $(A^{\top})^{\otimes k} = (A^{\otimes k})^{\top}$.
\end{enumerate} In addition, it has an important property that we will be using throughout this paper, the mixed product property.

\begin{lemma}[Mixed Product Property]
\label{lem: mixed product property}
If matrices $A,B,C,D$ are of such size such that one can form matrix products $AC,BD$, then
\[
(A \otimes B)(C \otimes D) = (AC) \otimes (BD).
\]
\end{lemma}

Now we can define the tensor product of (bi)linear algorithms, which is the key to multiply bigger matrices with small bilinear algorithms.

\begin{definition}[Tensor product of linear algorithms]
Let $A_1 = (V_1,E_1,\lambda_1)$ and $A_2 = (V_2,E_2,\lambda_2)$ be linear algorithms over $\mathbb{R}$. The tensor product of $A_1$ and $A_2$, denoted $A_1\otimes A_2$, is the linear algorithm formed from taking the graph tensor product of $(V_1,E_1,\lambda_1)$ and $(V_2,E_2,\lambda_2)$. In other words, $A_1\otimes A_2 = (V_1 \times V_2, E, \lambda)$ such that $E = \{((u_1,u_2),(v_1,v_2)):(u_1,v_1) \in E_1,(u_2,v_2)\in E_2\}$ and $\lambda((u_1,u_2),(v_1,v_2)) = \lambda_1(u_1,v_1)\cdot \lambda_2(u_2,v_2)$ for all $((u_1,u_2),(v_1,v_2)) \in E$. The tensor product of multiple linear algorithms is defined analogously.
\end{definition}

\begin{lemma}
\label{lem: tensor product of linear algorithm matrices}
Let $A_1,A_2,\ldots,A_k$ be linear algorithms, then $M_{\otimes_{i=1}^{k}A_i} = \bigotimes_{i=1}^{k}M_{A_{i}}$.  
\end{lemma}
\begin{proof}
It suffices to prove the statement when $k = 2$ since the statement then follows by induction. Suppose $A_1 = (V_1,E_1,\lambda_1)$, $|V_1^I| = t_1, |V_1^O| = s_1$ and $A_2 = (V_2,E_2,\lambda_2)$, $|V_2^I| = t_2, |V_2^O| = s_2$. By definition, $M_{A_1}\mathbf{x}_{t_1} = F_{A_1}$, $M_{A_2}\mathbf{x}_{t_2} = F_{A_2}$ and $M_{A_1 \otimes A_2}\mathbf{x}_{t_1t_2} = F_{A_1\otimes A_2}$. Notice that in vector form, $F_{A_1 \otimes A_2} = F_{A_1}\otimes F_{A_2}$, so $(M_{A_1}\mathbf{x}_{t_1})\otimes (M_{A_2}\mathbf{x}_{t_2}) = M_{A_1\otimes A_2}\mathbf{x}_{t_1t_2}$. By the mixed product property, this is equivalent to $(M_{A_1}\otimes M_{A_2})\cdot (\mathbf{x}_{t_1}\otimes \mathbf{x}_{t_2}) = M_{A_1\otimes A_2}\mathbf{x}_{t_1t_2}$, which means $M_{A_1}\otimes M_{A_2} = M_{A_1 \otimes A_2}$.
\end{proof}

\begin{definition}[Tensor product of bilinear algorithms]
Let $B = (A_{1},A_{2},A_{3}), B' = (A'_{1},A'_{2},A'_{3})$ be bilinear algorithms over $\mathbb{F}$. The tensor product of $B$ and $B'$ is $B \otimes B' = (A_{1}\otimes A'_{1}, A_{2} \otimes A'_{2}, A_{3} \otimes A'_{3})$. The tensor product of multiple bilinear algorithms is defined analogously. 
\end{definition}

Suppose in the definition above we have $A_{ij} = (V_{ij},E_{ij},\lambda_{ij})$ for $1 \leq i \leq 2$ and $1 \leq j \leq 3$. Then
\[
|(V_{11}\otimes V_{21})^O| = |V_{11}^O|\cdot |V_{21}^O| = |V_{13}^I|\cdot |V_{23}^I| = |(V_{13} \otimes V_{23})^I|,
\] and similarly $|(V_{12} \otimes V_{22})^O| = |(V_{13} \otimes V_{23})^I|$. Therefore the tensor product of bilinear algorithms is well-defined.

The tensor product of bilinear algorithms allow us to multiply bigger matrices. Given a bilinear algorithm $B$ for multiplying $n\times m$ and $m\times d$ matrices, we can use $B^{\otimes k}$ to multiply $n^k \times m^k$ and $m^k \times d^k$ matrices. 

\begin{lemma}
\label{lem: tensor product of bilinear algorithms}
Suppose $B = (A_1,A_2,A_3), B' = (A_1',A_2',A_3')$, where $A_i = (V_i,E_i,\lambda_i),A_j' = (V_j',E_j',\lambda_j')$, are both bilinear algorithms for matrix multiplication, then $B \otimes B'$ is also a bilinear algorithm for matrix multiplication with encoding matrices $M_{A_1}\otimes M_{A_1'},M_{A_2}\otimes M_{A_2'}$ and decoding matrix $M_{A_3}\otimes M_{A_3'}$. More generally, this is also true for multiple tensor products of bilinear algorithms. 
\end{lemma}
\begin{proof}
It suffices to prove the first statement and the second statement follows by induction. The first statement follows from $B \otimes B' = (A_1\otimes A_1',A_2\otimes A_2',A_3\otimes A_3')$ and that \Cref{lem: tensor product of linear algorithm matrices} shows that the encoding matrices are exactly $M_{A_1}\otimes M_{A_1'},M_{A_2}\otimes M_{A_2'}$ and the decoding matrix is exactly $M_{A_3}\otimes M_{A_3'}$.
\end{proof}

\section{Arithmetic Complexity Upper Bounds for Bilinear Algorithms}

From \Cref{lem: tensor product of bilinear algorithms} we know that if we have a bilinear algorithm $B$ for multiplying $n\times m$ and $m \times d$ matrices, then we can use it to construct a bilinear algorithm $B^{\otimes k}$ for multiplying $n^k \times m^k$ and $m^k \times d^k$ matrices. In this section, we present several upper bounds on the arithmetic complexity of tensor products of bilinear algorithms that will be useful later.

\subsection{Bilinear Algorithms Upper Bounds}

\begin{lemma}
\label{lem: cost of product}
Suppose $M = \prod_{i=1}^{n}M_i$, then $T(M) \leq \sum_{i=1}^{n}T(M_i)$.
\end{lemma}
\begin{proof}
This follows from definition. For any vector $v$, in order to compute $Mv$, we can simply multiply it by $M_{n},M_{n-1},\ldots,M_1$ sequentially, incurring $\sum_{i=1}^{n}T(M_i)$ arithmetic operations.
\end{proof}

\begin{lemma}
\label{lem: kronecker permutation does not change cost}
Given matrices $M_1,\ldots,M_n$ and any permutation $\sigma \in S_n$,
\[
T\Big(\bigotimes_{i=1}^{n} M_i\Big) = T\Big(\bigotimes_{i=1}^{n} M_{\sigma(i)}\Big).
\] In other words, the order of matrices in the Kronecker powers does not affect the number of operations needed for it to multiply a vector.
\end{lemma}
\begin{proof}
We know $\bigotimes_{i=1}^{n} M_i$ and $\bigotimes_{i=1}^{n} M_{\sigma(i)}$ are row/column permutations of each other. As a result, there exists two permutation matrices $P,Q$ of appropriate size such that $P(\bigotimes_{i=1}^{n} M_i)Q = \bigotimes_{i=1}^{n} M_{\sigma(i)}$. By \Cref{lem: cost of product}, we have 
\[
T\Big(\bigotimes_{i=1}^{n} M_{\sigma(i)}\Big) \leq T\Big(\bigotimes_{i=1}^{n} M_i\Big)+T(P)+T(Q) = T\Big(\bigotimes_{i=1}^{n} M_i\Big),
\] where we use the fact that it takes no operations to multiply a permutation matrix with a vector. We can also write $\bigotimes_{i=1}^{n} M_i = P^{-1}(\bigotimes_{i=1}^{n} M_{\sigma(i)})Q$ and the other direction is similar.
\end{proof}

We can use the mixed product property (\Cref{lem: mixed product property}) to give a generic bound of the arithmetic operations needed for $M^{\otimes k}$. The following lemmas (\Cref{lem: kronecker power upper bound}, \Cref{lem: multiplying big matrices with small tensors}) are the same as \cite[Claim 3.9]{BS19}, but we write them in a slightly different way for our purpose. 
\begin{lemma}
\label{lem: kronecker power upper bound}
For any $s\times t$ matrix $M$ and a positive integer $k$: 
\begin{itemize}
    \item If $s>t$, then $T(M^{\otimes k}) \leq T(M)\cdot \frac{s^k-t^k}{s-t}$;
    \item If $s<t$, then $T(M^{\otimes k}) \leq T(M)\cdot \frac{t^k-s^k}{t-s}$;
    \item If $s = t$, then $T(M^{\otimes k}) \leq T(M)\cdot kt^{k-1}$.
\end{itemize}
\end{lemma}
\begin{proof}
Suppose $s>t$, then using \Cref{lem: mixed product property} we can write 
\begin{align*}
    M^{\otimes k} &= (M^{\otimes k-1}\otimes I_{s})\cdot (I_{t^{k-1}}\otimes M)\\
    &= \Big(\big((M^{\otimes k-2}\otimes I_s)\cdot (I_{t^{k-2}}\otimes M)\big)\otimes I_s\Big)\cdot (I_{t^{k-1}}\otimes M)\\
    &= (M^{\otimes k-2}\otimes I_{s^2})\cdot (I_{s\cdot t^{k-2}}\otimes M)\cdot (I_{t^{k-1}}\otimes M).
\end{align*} An induction argument shows that 
\[
M^{\otimes k} = \prod_{i=1}^{k} (I_{t^{k-i}\cdot s^{i-1}} \otimes M).
\] Notice that we have abused notation in reordering terms in Kronecker products, since we know by \Cref{lem: kronecker permutation does not change cost} that these reorderings have the same cost. As a result, 
\begin{align*}
    T(M^{\otimes k}) &\leq \sum_{i=1}^{k}T(I_{t^{k-i}\cdot s^{i-1}} \otimes M)\\
    &\leq \sum_{i=1}^{k}T(M)\cdot t^{k-i}\cdot s^{i-1}\\
    &= T(M)\cdot \frac{s^k-t^k}{s-t}.
\end{align*} The case when $s = t$ can be computed in a similar way. When $s<t$, by \Cref{lem: matrix inverse time} we have 
\begin{align*}
    T(M^{\otimes k}) &= T((M^{\top})^{\otimes k})+(t^k-s^k)\\
    &= T(M^{\top})\cdot \frac{t^k-s^k}{t-s}+(t^k-s^k)\\
    &= (T(M)-(t-s))\cdot \frac{t^k-s^k}{t-s}+(t^k-s^k)\\
    &= T(M)\cdot \frac{t^k-s^k}{t-s}.
\end{align*}
\end{proof}

\Cref{lem: kronecker power upper bound} can be considered as recursively applying a linear algorithm for multiplying $M$ with a vector on $M^{\otimes k}$. This can be generalized to bilinear algorithms for multiplying matrices. 

\begin{lemma}
\label{lem: multiplying big matrices with small tensors}
Suppose there is a bilinear algorithm $B = (A_1,A_2,A_3)$ for computing $\langle n,m,d\rangle$, then for any positive integer $k$ we have
\[
T(n^k,m^k,d^k) \leq T(A_1)\cdot \frac{R(B)^k-(nm)^k}{R(B)-nm}+T(A_2)\cdot \frac{R(B)^k-(md)^k}{R(B)-md}+T(A_3)\cdot \frac{R(B)^k-(nd)^k}{R(B)-nd}+R(B)^k.
\]
\end{lemma}
\begin{proof}
By \Cref{lem: tensor product of bilinear algorithms} we know that $B^{\otimes k} = (A_1^{\otimes k},A_2^{\otimes k},A_3^{\otimes k})$ is a bilinear algorithm for multiplying $n^k \times m^k$ matrices with $m^k \times d^k$ matrices. Using \Cref{lem: kronecker power upper bound}, we can upper bound $T(n^k,m^k,d^k)$ by 
\begin{align*}
    &T(A_1^{\otimes k})+T(A_2^{\otimes k})+T(A_3^{\otimes k}) + R(B)^k\\
    &\leq T(A_1)\cdot \frac{R(B)^k-(nm)^k}{R(B)-nm}+T(A_2)\cdot \frac{R(B)^k-(md)^k}{R(B)-md}+T(A_3)\cdot \frac{R(B)^k-(nd)^k}{R(B)-nd}+R(B)^k.
\end{align*}
\end{proof}

\subsection{Multiplying Matrices Simultaneously}
\label{sec: multiplying matrices simultaneously}

One crucial idea we will use in this paper is to extend \Cref{lem: multiplying big matrices with small tensors} to simultaneously multiply $H$ copies of $n \times m$ and $m \times d$ matrices. We use $H\odot \mathcal{T}$ to denote the direct sum of $H$ copies of $\mathcal{T}$, and use $T(H\odot (n,m,d))$ to denote the arithmetic complexity of simultaneously multiplying $H$ copies of $n\times m$, $m \times d$ matrices. Notice that we trivially have\footnote{Note that this also holds true if $T$ is replaced with rank. Interestingly, it was shown that this inequality is sometimes strict \cite{Shitov19}.}
\[
T(H\odot (n,m,d)) \leq H\cdot T(n,m,d).
\] 

One commonly used method of upper bounding the rank of multiple copies of kronecker powers is the Sch\"{o}nhage Theorem.

\begin{lemma}[\cite{Sch}, see also \cite{Blaser}]
\label{lem: Schonhage thm}
If $H,r,n,m,d$ are positive integers and $R(H\odot \langle n,m,d\rangle) \leq r$, then for any positive integer $s$, $R(H\odot \langle n^s,m^s,d^s\rangle)\leq \lceil r/H\rceil^s\cdot H$. In particular, $R(\langle n^s,m^s,d^s\rangle) \leq \lceil r/H\rceil^s\cdot H$.
\end{lemma}

In general, $H\odot \langle n,m,d\rangle$ does not have to be a matrix multiplication tensor. However, in some cases, it can be reduced from some matrix multiplication tensor (for example, $\CW$ tensor, as we will see soon) via zeroing out. That allows us to use the encoding/decoding matrices from that matrix multiplication tensor to compute $H\odot \langle n,m,d\rangle$. 

The main reason why we want to multiply matrices simultaneously is because there is no good upper bound for the rank of the $\CW$ tensor. However, a good \emph{border rank} exists for $\CW$ tensor, which will allow us to construct $\ell$ many tensors $T_j$ with small rank such that the sum of $T_j$ can be zeroed out to $H\odot \langle n,m,d\rangle$.

For simplicity, we will assume that $H$ divides $R(\mathcal{T})$ below. This is generally not true, but one can always pad zeros to matrices/tensors without increasing the size too much. This blowup will be very small and will not affect the final outcome.

\begin{lemma}
\label{lem: simultaneous recursion sum}
Suppose $H,n,m,d,\ell$ are positive integers, $\mathcal{T}_j,\mathcal{T}$ are tensors such that 
\[
H\odot \langle n,m,d\rangle \leq \sum_{j=1}^{\ell}\mathcal{T}_j = \mathcal{T}.
\] Let $B_j = (A_{j1},A_{j2},A_{j3})$ be any bilinear algorithm computing $\mathcal{T}_j$ and let $t = \sum_{j=1}^{\ell}R(B_j)$ such that $\frac{t}{H} \geq 2\max\{nm,md,nd\}$. Then for any positive integer $k$ we can upper bound $T(H\odot \langle n^k,m^k,d^k\rangle)$ with
\begin{align*}
    \Big(\frac{t}{H}\Big)^k\cdot H+\Big(\frac{t}{H}\Big)^{k-1}\cdot \Big(\sum_{j=1}^{\ell}\Big(T(A_{j1})+T(A_{j2})+T(A_{j3})\Big)+(\ell-1)|\mathcal{T}|\Big).
\end{align*} 
\end{lemma}
\begin{proof}
To multiply $H$ copies of $(n^k,m^k,d^k)$ matrix multiplications $P_i \cdot Q_i$, we partition each $P_i$ into $n \times m$ many $n^{k-1}\times m^{k-1}$ matrices and $Q_i$ into $m \times d$ many $m^{k-1}\times d^{k-1}$ matrices and use $H\odot \langle n,m,d\rangle$ to recursively compute their products. Since $B_j = (A_{j1},A_{j2},A_{j3})$ is a bilinear algorithm computing $\mathcal{T}_j$, we can obtain a bilinear algorithm $B$ computing $\mathcal{T}$ by concatenating $A_{j1},A_{j2}$'s by rows and $A_{j3}$ by columns, i.e. the first encoding matrix of $B$ is the concatenation of all $A_{j1}$ by rows; the second encoding matrix of $B$ is the concatenation of all $A_{j2}$ by rows; the decoding matrix is the concatenation of all $A_{j3}$ by columns. We will denote this bilinear algorithm by $B = (A_1,A_2,A_3)$. Now $T(H\odot (n^k,m^k,d^k))$ can be upper bounded by
\begin{align*}
   &T\Big(t\odot  (n^{k-1},m^{k-1},d^{k-1})\Big)+T(A_1)\cdot n^{k-1}m^{k-1}+T(A_2)\cdot m^{k-1}d^{k-1}+T(A_3)\cdot n^{k-1}d^{k-1}\\
   &\leq \frac{t}{H}\cdot  T\Big(H\odot (n^{k-1},m^{k-1},d^{k-1})\Big)+T(A_1)\cdot n^{k-1}m^{k-1}+T(A_2)\cdot m^{k-1}d^{k-1}+T(A_3)\cdot n^{k-1}d^{k-1}.
\end{align*}
Since $R(H\odot \langle n,m,d\rangle)>H$, we always have enough (more than $H$) copies to multiply simultaneously. Solving this recursion gives us the final upper bound of $T(H\odot (n^k,m^k,d^k))$ as
\begin{align*}
    &\Big(\frac{t}{H}\Big)^k\cdot H+\sum_{i=1}^{k}\Big(T(A_1)\cdot n^{k-i}m^{k-i}\cdot \Big(\frac{t}{H}\Big)^{i-1}+T(A_2)\cdot m^{k-i}d^{k-i}\cdot \Big(\frac{t}{H}\Big)^{i-1}+T(A_3)\cdot n^{k-i}d^{k-i}\cdot \Big(\frac{t}{H}\Big)^{i-1}\Big)\\
    &= \Big(\frac{t}{H}\Big)^k\cdot H+T(A_1)\cdot \frac{(t/H)^k-(mn)^k}{(t/H)-mn}+T(A_2)\cdot \frac{(t/H)^k-(md)^k}{(t/H)-md}+T(A_3)\cdot \frac{(t/H)^k-(nd)^k}{(t/H)-nd}\\
    &\leq \Big(\frac{t}{H}\Big)^k\cdot H+\Big(T(A_1)+T(A_2)+T(A_3)\Big)\cdot \Big(\frac{t}{H}\Big)^{k-1},
\end{align*} where we use the assumption that $\frac{t}{H}\geq \max\{2mn,2nd,2md\}$. Finally, notice that since we concatenate the linear algorithms, we have
\[
T(A_1) \leq \sum_{j=1}^{\ell}T(A_{j1}), T(A_2) \leq \sum_{j=1}^{\ell}T(A_{j2})
\] and 
\[
T(A_3) \leq \sum_{j=1}^{\ell}T(A_{j3})+(\ell-1)|\mathcal{T}|
\] where the extra $(\ell-1)|\mathcal{T}|$ term comes from summing up all $\ell$ values computed for each coefficient.
\end{proof}

\section{Better Leading Constant via Rectangular MM}
\label{sec: better leading constants via rectangular MM}

In this section we present our main result (\Cref{thm: main result}): Given a matrix multiplication tensor $\langle n,n,n\rangle$ with rank at most $t$, we show that if $k>O((\log n)^{\varepsilon})$ is an integer where $\varepsilon>0$ is a sufficiently small constant, we have
\[
T(n^k,n^k,n^k) \leq n^{O(1/(\log n)^{\frac{1}{3}-\varepsilon})}\cdot t^k.
\] To our knowledge, all previous bilinear algorithms using large tensor $\langle n,n,n\rangle$ will incur a much larger cost, sometimes up to $3n^2t^k$, so we improve the leading constant from $3n^2$ to $n^{o(1)}$.

\subsection{Matrix Multiplication via Rectangular Matrix Multiplication}

We improve the leading constant by using the fact that we can multiply rectangular matrix multiplications slightly faster. We first show that we can reduce square matrix multiplication into rectangular matrix multiplication.

\begin{lemma}
\label{lem: rectangular matrix multiplication}
Suppose $R(\langle n,n,n\rangle) \leq t$ for some positive integers $n,t$, then for any positive integer $k \geq 1$ we have
\begin{align*}
T(n^k,n^k,n^k) &\leq t^k+\frac{4t^{k-1}}{n^{2(k-1)}}\cdot T(t,n^2,n^{2(k-1)})+\frac{2t^{k-1}}{n^{2(k-1)}}\cdot T(n^2,t,n^{2(k-1)}).
\end{align*}
\end{lemma}
\begin{proof}
Let $B = (A_1,A_2,A_3)$ be a bilinear algorithm for $\langle n,n,n\rangle$ with $R(B) \leq t$. Since $R(B) \leq t$, we can assume that $M_{A_1},M_{A_2}$ are $t \times n^2$ matrices and $M_{A_3}$ is an $n^2 \times t$ matrix. By \Cref{lem: tensor product of bilinear algorithms} we know that $B^{\otimes k} = (A_1^{\otimes k},A_2^{\otimes k},A_3^{\otimes k})$ is a bilinear algorithm for $\langle n^k,n^k,n^k\rangle$ and that $M_{A_i^{\otimes k}} = M_{A_i}^{\otimes k}$ for all $1 \leq i \leq 3$. Therefore, we obtain
\[
T(n^k,n^k,n^k) \leq t^k+T(M_{A_1}^{\otimes k})+T(M_{A_2}^{\otimes k})+T(M_{A_3}^{\otimes k}).
\]

Using (the proof of) \Cref{lem: kronecker power upper bound}, we know for $j = 1,2$ we have
\[
M_{A_j}^{\otimes k} = \prod_{i=1}^{k} (I_{n^{2(k-i)}\cdot t^{i-1}} \otimes M_{A_j}).
\] (We have abused notation in reordering terms in Kronecker products, since we know by \Cref{lem: kronecker permutation does not change cost} that these reorderings have the same cost.) Therefore, for $j = 1,2$ we have
\[
T(M_{A_j}^{\otimes k}) \leq \sum_{i=1}^{k}T(I_{n^{2(k-i)}\cdot t^{i-1}} \otimes M_{A_j}).
\] However, notice that multiplying $I_{n^{2(k-i)\cdot t^{i-1}}} \otimes M_{A_j}$ with a vector of length $n^{2(k-i+1)}\cdot t^{i-1}$ is equivalent to separately multiplying $M_{A_j}$ by $n^{2(k-i)}\cdot t^{i-1}$ different sub-vectors of $v$ of length $n^2$, which is in turn equivalent to multiplying $M_{A_j}$ times a matrix of dimensions $n^2 \times (n^{2(k-i)})\cdot t^{i-1})$. Therefore for $j = 1,2$, we get
\begin{align*}
    T(M_{A_j}^{\otimes k}) &\leq \sum_{i=1}^{k}T(t,n^2,n^{2(k-i)}\cdot t^{i-1})\\
    &\leq \sum_{i=1}^{k}T(t,n^2,n^{2(k-1)})\cdot \frac{n^{2(k-i)}\cdot t^{i-1}}{n^{2(k-1)}}\\
    &= T(t,n^2,n^{2(k-1)})\cdot \sum_{i=1}^{k}\frac{n^{2(k-i)}\cdot t^{i-1}}{n^{2(k-1)}}\\
    &= T(t,n^2,n^{2(k-1)})\cdot \frac{t^k-n^{2k}}{t-n^2}\cdot \frac{1}{n^{2(k-1)}}\\
    &\leq T(t,n^2,n^{2(k-1)})\cdot \frac{2t^{k-1}}{n^{2(k-1)}},
\end{align*} where the last inequality uses the fact that $t \geq 2n^2$ for any matrix multiplication tensor $\langle n,n,n\rangle$ with rank $t$ (\cite{landsberg2014new}) such that
\[
t^k-n^{2k} \leq 2t^{k-1}(t-n^2) = 2t^k-2n^2t^{k-1}.
\]

As for $M_{A_3}^{\otimes k}$, we can similarly use the mixed product property (\Cref{lem: mixed product property}) to yield
\begin{align*}
    M_{A_3}^{\otimes k} &= (M_{A_3}^{\otimes k-1}\otimes I_{n^2})(I_{t^{k-1}}\otimes M_{A_3})\\
    &= \Big(\big((M_{A_3}^{\otimes k-2}\otimes I_{n^2})\cdot (I_{t^{k-2}}\otimes M_{A_3})\big)\otimes I_{n^2}\Big)\cdot (I_{t^{k-1}}\otimes M_{A_3})\\
    &= (M_{A_3}^{\otimes k-2}\otimes I_{n^4})\cdot (I_{t^{k-2}\cdot n^2}\otimes M_{A_3})\cdot (I_{t^{k-1}}\otimes M_{A_3}).
\end{align*} An induction argument shows that
\[
M_{A_3}^{\otimes k} = \prod_{i=1}^{k}(I_{n^{2(i-1)}\cdot t^{k-i}}\otimes M_{A_3}).
\] Again we have 
\[
T(M_{A_3}^{\otimes k}) \leq \sum_{i=1}^{k}T(n^2,t,n^{2(k-i)}\cdot t^{i-1}),
\] which yields
\[
T(M_{A_3}^{\otimes k}) \leq T(\langle n^2,t,n^{2(k-1)}\rangle)\cdot \frac{2t^{k-1}}{n^{2(k-1)}}.
\] Finally, we sum all the costs up (we also have a length-$t^k$ vector-wise multiplication) to get
\begin{align*}
T(n^k,n^k,n^k) &\leq t^k+\frac{4t^{k-1}}{n^{2(k-1)}}\cdot T(t,n^2,n^{2(k-1)})+\frac{2t^{k-1}}{n^{2(k-1)}}\cdot T(n^2,t,n^{2(k-1)}).
\end{align*}
\end{proof}

Thus, in order to design a fast algorithm for \emph{square} matrix multiplication, it suffices to use a fast algorithm for \emph{very rectangular} matrix multiplications. We would also like to point out that in general $T(t,n^2,n^{2(k-1)})$ and $T(n^2,t,n^{2(k-1)})$ could be different, even though the trivial algorithm gives the exact same arithmetic complexity when multiplying a $t \times n^2$ matrix with a $n^2 \times n^{2(k-1)}$ matrix and a $n^2 \times t$ matrix with a $t \times n^{2(k-1)}$ matrix.

Matrix multiplication of the form $\langle q^a,q^b,q^c\rangle$ for constants $a,b,c$ is well-studied~\cite{HP98}. Toward this end, we will prove (in \Cref{sec: proof of main theorem}):

\begin{theorem}
\label{thm: rectangular matrix multiplication}
Given a tensor $\langle n,n,n\rangle$ with rank at most $t$ and a sufficiently small constant $\varepsilon>0$, for any $k \geq \Omega((\log n)^{\varepsilon})$,
\begin{align*}
    T(t,n^2,n^{2(k-1)}) &\leq n^{O(1/(\log n)^{\frac{1}{3}-\varepsilon})}\cdot n^{2(k-1)}\cdot t,\\
    T(n^2,t,n^{2(k-1)}) &\leq n^{O(1/(\log n)^{\frac{1}{3}-\varepsilon})}\cdot n^{2(k-1)}\cdot t.
\end{align*}

\end{theorem}

We can plug \Cref{thm: rectangular matrix multiplication} into \Cref{lem: rectangular matrix multiplication} to obtain our main result.

\begin{theorem}
\label{thm: main result}
Given a tensor $\langle n,n,n\rangle$ with rank at most $t$ and a sufficiently small constant $\varepsilon>0$, for any $k \geq \Omega((\log n)^{\varepsilon})$,
\[
T(n^k,n^k,n^k) \leq n^{O(1/(\log n)^{\frac{1}{3}-\varepsilon})}\cdot t^{k}.
\]
\end{theorem}
\begin{proof}
It follows from \Cref{thm: rectangular matrix multiplication} and \Cref{lem: rectangular matrix multiplication} that
\begin{align*}
    T(n^k,n^k,n^k) &\leq t^k+\frac{4t^{k-1}}{n^{2(k-1)}}\cdot T(t,n^2,n^{2(k-1)})+\frac{2t^{k-1}}{n^{2(k-1)}}\cdot T(n^2,t,n^{2(k-1)})\\
    &\leq t^k+\frac{6t^{k-1}\cdot t}{n^{2(k-1)}}\cdot n^{O(1/(\log n)^{\frac{1}{3}-\varepsilon})}\cdot n^{2(k-1)}\\
    &\leq n^{O(1/(\log n)^{\frac{1}{3}-\varepsilon})}\cdot t^{k}.
\end{align*}
\end{proof}

\subsection{Rectangular Matrix Multiplication}
\label{sec: rectangular mm}

Now we will give our algorithms for multiplying rectangular matrix multiplications which will eventually prove \Cref{thm: rectangular matrix multiplication}. We will mostly use the method described in \cite[Section 7.3]{HP98} for rectangular matrix multiplication with size $n^a \times n^b \times n^c$, but we will particularly need to make use of algorithms in the case when $a,b,c$ are growing with $n$. We will use the following Coopersmith-Winograd tensor \cite{CW87} for rectangular matrix multiplication. 
\[
\CW_q = \sum_{i=1}^q (x_0^{[0]} y_i^{[1]} z_i^{[1]} + x_i^{[1]} y_0^{[0]} z_i^{[1]} + x_i^{[1]} y_i^{[1]} z_0^{[0]})+x_0^{[0]}y_0^{[0]}z_{q+1}^{[2]}+x_0^{[0]}y_{q+1}^{[2]}z_{0}^{[0]}+x_{q+1}^{[2]}y_{0}^{[0]}z_{0}^{[0]}.
\] Notice that $q$ is a positive integer to be determined. Each variable $x_i^{[i']}$ is a variable such that $i$ is chosen from $\{0,\ldots,q+1\}$ and $i' \in \{0,1,2\}$ depends on $i$ such that: if $i = 0$, then $i' = 0$; if $i \in \{1,\ldots,q\}$, then $i' = 1$; if $i = q+1$, then $i' = 2$. Variables $y,z$ are defined similarly.

\begin{lemma}
\label{lem: simpleCW}
    For any positive real numbers $a,b,c$ and positive integer $q$ such that $a\leq b\leq c$, there exists sufficiently large positive integers $N,L_1,L_2,L_3$ such that if
    \[
    {(b+c)N \choose bN,cN}\cdot {aN+L_1+L_3 \choose aN,L_1,L_3}
    \] is greater than both
    \[
    {(a+c)N \choose aN,cN}\cdot {bN+L_2+L_3 \choose bN,L_2,L_3},{(a+b)N \choose aN, bN}\cdot {cN+L_1+L_2\choose cN,L_1,L_2},
    \] then the tensor $H \odot \langle q^{aN},q^{bN},q^{cN}\rangle \leq \CW_q^{\otimes (a+b+c)N+L_1+L_2+L_3}$ with rank $t$, where 
    \begin{align*}
        H &\geq \frac{1}{2^{O(\sqrt{(a+b+c)N+L_1+L_2+L_3})}}\cdot \frac{{(a+b+c)N+L_1+L_2+L_3 \choose L_1,L_2,L_3,aN,bN,cN}}{2{(b+c)N \choose bN,cN} \cdot{aN+L_1+L_3 \choose aN,L_1,L_3}+1}
    \end{align*} and $t \leq O\Big((q+2)^{(a+b+c)N+L_1+L_2+L_3}\cdot ((a+b+c)N+L_1+L_2+L_3)^2\Big)$.
\end{lemma}
\begin{proof}
We will show this follows from the proof in \cite[Section 7.3]{HP98}. The differences between our arguments and theirs are that we consider the case when the exponents $aN,bN,cN$ are growing, and we also spread weights differently for the terms in the Coopersmith-Winograd tensor.

For a positive integer $N$ to be determined, consider the Kronecker power $\CW_q^{\otimes (a+b+c)N+L_1+L_2+L_3}$. Each variable $x_{\alpha}^{I}$ in the tensor power is a tensor product of the initial variables $x_{i}^{[i']}$ where $i \in \{0,1,\ldots,q\}, i' \in \{0,1,2\}$ (similarly for $y,z$). The superscript $I$ is a vector of length $(a+b+c)N+L_1+L_2+L_3$ with entries in $\{0,1,2\}$ and the subscript $\alpha$ is a vector of length $(a+b+c)N+L_1+L_2+L_3$ with entries in $\{0,1,\ldots,q\}$.

We begin with the border rank construction for $\CW_q$ (\cite{CW87}, using the notation of \cite[Section 3.1]{HP98}), which shows that $\CW_q$ has border rank at most $q+2$:
\begin{equation}
\label{eq: CW tensor}
    \begin{split}
        \CW_q = &\sum_{i=1}^q (x_0^{[0]} y_i^{[1]} z_i^{[1]} + x_i^{[1]} y_0^{[0]} z_i^{[1]} + x_i^{[1]} y_i^{[1]} z_0^{[0]})+x_0^{[0]}y_0^{[0]}z_{q+1}^{[2]}+x_0^{[0]}y_{q+1}^{[2]}z_{0}^{[0]}+x_{q+1}^{[2]}y_{0}^{[0]}z_{0}^{[0]} + O(\lambda)
        \\ &= \sum_{i=1}^q \lambda^{-2} (x_0^{[0]} + \lambda x_i^{[1]})(y_0^{[0]} + \lambda y_i^{[1]})(z_0^{[0]} + \lambda z_i^{[1]})
        \\ &- \lambda^{-3} \left( x_0^{[0]} + \lambda^2 \sum_{i=1}^q x_i^{[1]} \right) \left( y_0^{[0]} + \lambda^2 \sum_{i=1}^q y_i^{[1]} \right) \left( z_0^{[0]} + \lambda^2 \sum_{i=1}^q z_i^{[1]} \right)
        \\ &+ [\lambda^{-3} - q \lambda^{-2}] (x_0^{[0]}+\lambda^3x_{q+1}^{[2]})(y_0^{[0]}+\lambda^3y_{q+1}^{[2]})(z_0^{[0]}+\lambda^3z_{q+1}^{[2]}).
    \end{split}
\end{equation}
Note in particular that the degree of $\lambda$ on each side is a constant independent of $q$. It follows by~\cite{bini} (see also~\cite[Lemma 6.4]{Blaser}) that $\CW_q^{\otimes(a+b+c)N+L_1+L_2+L_3}$ has rank at most 
\[
O\Big((q+2)^{(a+b+c)N+L_1+L_2+L_3}\cdot ((a+b+c)N+L_1+L_2+L_3)^2\Big).
\]

Let us now determine the ``value" of $\CW_q^{\otimes (a+b+c)N+L_1+L_2+L_3}$.  
\begin{itemize}
    \item Zero out all $x_{\alpha}^{I}$ unless $I$ has $cN+L_1+L_2$ indices equal to $0$, $(a+b)N$ indices equal to $1$ and $L_3$ indices equal to $2$.
    \item Zero out all $y_{\beta}^{J}$ unless $J$ has $aN+L_1+L_3$ indices equal to $0$, $(b+c)N$ indices equal to $1$ and $L_2$ indices equal to $2$. 
    \item Zero out all $z_{\gamma}^{K}$ unless $K$ has $bN+L_2+L_3$ indices equal to $0$, $(a+c)N$ indices equal to $1$ and $L_1$ indices equal to $2$.
\end{itemize}

At the conclusion of this procedure, our tensor is a sum of 
\[
{(a+b+c)N+L_1+L_2+L_3 \choose L_1,L_2,L_3,aN,bN,cN}
\] blocks of triples $(X^I,Y^J,Z^K)$ which are of the form $\langle q^{aN},q^{bN},q^{cN}\rangle$. Among them, 
\begin{itemize}
    \item for each $Z^K$ there are 
    \[
    {(a+c)N \choose aN,cN}\cdot {bN+L_2+L_3\choose bN,L_2,L_3}
    \] pairs of $(X^I,Y^J)$ sharing it;
    \item for each $Y^J$ there are 
    \[
    {(b+c)N \choose bN,cN} \cdot{aN+L_1+L_3 \choose aN,L_1,L_3}
    \] pairs of $(X^I,Z^K)$ sharing it;
    \item for each $X^I$, there are
    \[
    {(a+b)N \choose aN,bN}\cdot {cN+L_1+L_2\choose cN,L_1,L_2}
    \] pairs of $(Y^J,Z^K)$ sharing it.
\end{itemize} By our assumption, the second term is the largest of these three quantities, so we set
\[
M := 2{(b+c)N \choose bN,cN} \cdot{aN+L_1+L_3 \choose aN,L_1,L_3}+1.
\] Now let $B \subset \mathbb{Z}_M$ be a Salem-Spencer set (no three-term arithmetic progressions) by \cite{SS42,B46}, we pick such a set of size $M' \geq M^{1-\gamma/\sqrt{\log M}}$ for some universal constant $\gamma>0$. As in \cite{HP98}, we can hash variables to $\mathbb{Z}_M$ to further zero out our tensor to a direct sum of $H$ copies of $\langle q^{aN},q^{bN},q^{cN}\rangle$ which do not share any variables with each other, where
\begin{align*}
H &= \frac{1}{4}\cdot \frac{M'}{M^2}\cdot {(a+b+c)N+L_1+L_2+L_3 \choose L_1,L_2,L_3,aN,bN,cN}\\
&= \frac{1}{2^{O(\sqrt{\log M})}}\cdot \frac{{(a+b+c)N+L_1+L_2+L_3 \choose L_1,L_2,L_3,aN,bN,cN}}{2{(b+c)N \choose bN,cN} \cdot{aN+L_1+L_3 \choose aN,L_1,L_3}+1}\\
&\geq \frac{1}{ 2^{O(\sqrt{(a+b+c)N+L_1+L_2+L_3})}}\cdot \frac{{(a+b+c)N+L_1+L_2+L_3 \choose L_1,L_2,L_3,aN,bN,cN}}{2{(b+c)N \choose bN,cN} \cdot{aN+L_1+L_3 \choose aN,L_1,L_3}+1}.
\end{align*} The last inequality holds because 
\begin{align*}
    2{(b+c)N \choose bN,cN}\cdot {aN+L_1+L_3 \choose aN,L_1,L_3}+1\leq \frac{(aN+L_1+L_3)!\cdot ((b+c)N)!}{(bN)!(cN)!(aN)!L_1!L_3!} \leq 2^{(a+b+c)N+L_1+L_2+L_3}.
\end{align*}
\end{proof}

Eventually we will be choosing $a,b,c,L_1,L_2,L_3,N$, so here we will prove some properties for the parameters that we will pick.

\begin{lemma}
\label{lem: Josh flight}
    Let $a = 1, b = 1+\delta$ for some $\delta \in (0,1/2)$, $c = k$ for some super-constant $k>\Omega(1)$ and $P$ be a positive integer to be determined. Let $q = (b+c)(1-\delta/2)$, $\alpha = \frac{P}{q+2}$, $L_1 = L_3 = \frac{\delta\alpha}{4}$, $L_2 = \alpha, N = (1-\delta/2)\alpha$. Then we have 
    \[
    (a+b+c)N+L_1+L_2+L_3 = P,
    \] and 
    \[
    {(b+c)N \choose bN,cN}\cdot {aN+L_1+L_3 \choose aN,L_1,L_3}
    \] is greater than both
    \[
    {(a+c)N \choose aN,cN}\cdot {bN+L_2+L_3 \choose bN,L_2,L_3},{(a+b)N \choose aN, bN}\cdot {cN+L_1+L_2\choose cN,L_1,L_2}.
    \]
\end{lemma}
\begin{proof}
It is not hard to check the first equation:
\begin{align*}
    (a+b+c)N + L_1 + L_2 + L_3 &= ((1+1+\delta+k)(1-\delta/2)+\delta/4+\delta/4+1)\alpha\\
    &= (1 - \delta/2 + 1 - \delta/2+\delta - \delta^2/2 + k - k\delta/2 + \delta/4 + 1 + \delta/4)\alpha\\
    &= (3 + \delta/2 - \delta^2/2 + k - k \delta/2)\alpha \\
    &= (q+2)\alpha \\
    &= P.
\end{align*} For the second part, we denote for convenience that 
\[
Q_a = {(b+c)N \choose bN,cN}\cdot {aN+L_1+L_3 \choose aN,L_1,L_3}, Q_b = {(a+c)N \choose aN,cN}\cdot {bN+L_2+L_3 \choose bN,L_2,L_3},
Q_c = {(a+b)N \choose aN, bN}\cdot {cN+L_1+L_2\choose cN,L_1,L_2}
\] and we want to show $Q_a>Q_b,Q_c$. Further define 
\[
Q = {(a+b+c)N+L_1+L_2+L_3 \choose aN,bN,cN,L_1,L_2,L_3}
\] and notice that 
\begin{align*}
    \frac{Q}{Q_a} &= \binom{(a+b+c)N + L_1 + L_2 + L_3}{(b+c)N, aN+L_1+L_3, L_2}= \binom{P}{q\alpha,\alpha,\alpha},
    = \binom{P}{q\alpha, 2\alpha} \cdot \binom{2\alpha}{\alpha, \alpha}\\
    \frac{Q}{Q_b} &= \binom{(a+b+c)N + L_1 + L_2 + L_3}{(a+c)N, bN + L_2 + L_3, L_1}\\
    &= \binom{P}{(\delta/4)\alpha, (2 + 3\delta/4 - \delta^2 / 2)\alpha, (q-\delta - \delta^2/2)\alpha}\\
    &= \binom{P}{(q-\delta - \delta^2/2)\alpha, (2 + \delta + \delta^2/2)\alpha} \cdot \binom{(2 + \delta + \delta^2/2)\alpha}{(\delta/4)\alpha,(2 + 3\delta/4 - \delta^2 / 2)\alpha},\\
    \frac{Q}{Q_c} &= \binom{(a+b+c)N + L_1 + L_2 + L_3}{(a+b)N, cN + L_1 + L_2, L_3}\\
    &= \binom{P}{(2-\delta^2 / 2)\alpha, (\delta/4)\alpha, (q - \delta/4 + \delta^2/2)\alpha}\\
    &= \binom{P}{(q - \delta/4 + \delta^2/2)\alpha, (2 + \delta/4 - \delta^2/2)\alpha} \cdot \binom{(2 + \delta/4 - \delta^2/2)\alpha}{(2-\delta^2 / 2)\alpha, (\delta/4)\alpha}.
\end{align*} $Q/Q_a$ is the smallest since $q$ is super-constant, so in each of these final expressions as the product of two binomial coefficients, the first binomial coefficient is exponentially larger than the second, and the $Q/Q_a$ has the smallest first binomial coefficient. For instance, $\binom{P}{q\alpha, 2\alpha} \ll \binom{P}{(q-\delta - \delta^2/2)\alpha, (2 + \delta + \delta^2/2)\alpha}$ since $q\alpha > (q-\delta - \delta^2/2)\alpha > P/2$, from which it follows that $\frac{Q}{Q_a} \ll \frac{Q}{Q_b}$. (For $Q_c$, note that $\delta/4 - \delta^2/2 > 0$ since $\delta < 1/2$.)

Therefore, we conclude that $Q_a$ is the largest among the three quantities.
\end{proof}



In \Cref{thm: rectangular matrix multiplication} we are essentially trying to multiply an $m^{ak} \times m^{bk}$ matrix with an $m^{bk} \times m^{ck}$ matrix, where $\langle m^a,m^b,m^c\rangle$ is the base tensor that we use. We have shown in \Cref{lem: simultaneous recursion sum} that if we have $H$ copies of $\langle m^a,m^b,m^c\rangle$, then we can use it to multiply $m^{ak} \times m^{bk}$ and $m^{bk} \times m^{ck}$ matrices. However, since there is only one copy of matrix multiplication in the beginning, we need to do recursions using $\langle m^a,m^b,m^c\rangle$ to obtain at least $H$ copies. 

\begin{theorem}
\label{thm: the algorithm}
Assuming $H\odot \langle m^a,m^b,m^c\rangle \leq \sum_{i=1}^{\ell}\mathcal{T}_i = \mathcal{T}$ for some positive integers $H,a,b,c,\ell$ and tensors $\mathcal{T}_i,\mathcal{T}$. Let $r = R(\langle m^a,m^b,m^c\rangle)$, $B_i = (A_{i1},A_{i2},A_{i3})$ be a bilinear algorithm computing $\mathcal{T}_i$ and $t = \sum_{j=1}^{\ell}R(B_j)$ such that $\frac{t}{H} \geq 2\max\{m^{a+b},m^{a+c},m^{b+c}\}$. Then $T(m^{ak},m^{bk},m^{ck})$ can be upper bounded by
\begin{align*}
    &r\cdot \log_r H\cdot (m^{(a+b)k}+m^{(b+c)k}+m^{(a+c)k})+H\Big(\frac{t}{H}\Big)^{k-\log_r H}\\
    &+\Big(\frac{t}{H}\Big)^{k-\log_r H-1}\cdot \sum_{j=1}^{\ell}\Big(T(A_{j1})+T(A_{j2})+T(A_{j3})+(\ell-1)\cdot|\mathcal{T}|\Big)
\end{align*} for a sufficiently large positive integer $k$.
\end{theorem}
\begin{proof}
The idea is to run any bilinear algorithm for $\langle m^a,m^b,m^c\rangle$ for enough times such that there are at least $H$ copies of matrix multiplications, and then we will use $B_i$'s. 

We will run any bilinear algorithm for $\langle m^a,m^b,m^c\rangle$ recursively for $\log_r H$ times (again wlog we can assume this is an integer as using $\lceil \log_r H\rceil$ will not affect the time up to a constant factor). We obtain an upper bound for $T(m^{ak},m^{bk},m^{ck})$ as
\begin{align*}
    r\cdot \log_r H\cdot (m^{(a+b)k}+m^{(b+c)k}+m^{(a+c)k})+T\Big(H\odot (m^{a(k-\log_rH)},m^{b(k-\log_r H)},m^{c(k-\log_r H)})\Big),
\end{align*} where the first term is the number of additions needed and the second term is the number of multiplications needed. Now we use $B_i$'s to upper bound $T\Big(H\odot (m^{a(k-\log_rH)},m^{b(k-\log_r H)},m^{c(k-\log_r H)})\Big)$. By Lemma \ref{lem: simultaneous recursion sum} we have 
\begin{align*}
    &T\Big(H\odot (m^{a(k-\log_rH)},m^{b(k-\log_r H)},m^{c(k-\log_r H)})\Big) \\
    &\leq H\Big(\frac{t}{H}\Big)^{k-\log_r H}+\Big(\frac{t}{H}\Big)^{k-\log_r H-1}\cdot \sum_{j=1}^{\ell}\Big(T(A_{j1})+T(A_{j2})+T(A_{j3})+(\ell-1)\cdot |\mathcal{T}|\Big).
\end{align*} Combining the two inequalities gives us the final upper bound.
\end{proof}

\subsection{From Border Rank to Rank}

\Cref{thm: the algorithm} allows us to upper bound $T(m^{ak},m^{bk},m^{ck})$ if there exists $\mathcal{T}_i$ such that 
\[
H\odot \langle m^a,m^b,m^c\rangle \leq \sum_{i=1}^{\ell}\mathcal{T}_i.
\] As discussed in \Cref{sec: multiplying matrices simultaneously}, the reason why we need these $\mathcal{T}_i$ is that there is no good upper bound for the rank of $\CW_q$. However, we know the border rank of $\CW_q$ is at most $q+2$, which allows us to construct such $\mathcal{T}_i$'s with rank at most $q+2$. We construct them for both infinite field $\mathbb{C}$ and finite fields $\mathbb{F}$.

\begin{lemma}
\label{lem: from border rank to rank}
Suppose $\mathcal{T}$ is a tensor over $\mathbb{C}$ such that $\mathcal{U} \leq \mathcal{T}^{\otimes k}$ for some tensor $\mathcal{U}$ and sufficiently large positive integer $k$, then there exist tensors $\mathcal{T}_1,\ldots,\mathcal{T}_\ell$ which all have rank at most $\ubar{R}(\mathcal{T})$, such that $\mathcal{U} \leq \sum_{i=1}^{\ell} \mathcal{T}_i^{\otimes k}$ and $\ell \leq O(k)$.
\end{lemma}
\begin{proof}
Let $r = \ubar{R}(\mathcal{T})$ denote the border rank of $\mathcal{T}$. Thus, there exist tensors $\mathcal{A}_1,\ldots,\mathcal{A}_d$ for a constant $d$ such that, for any nonzero $\varepsilon$, the tensor $\mathcal{A}_\varepsilon$ defined as
\[
\mathcal{A}_\varepsilon = \mathcal{T}+\varepsilon^1\mathcal{A}_1+\ldots+\varepsilon^d\mathcal{A}_d
\] is a tensor with rank at most $r$. Let $dk \leq p \leq 2dk$ be a prime number (which exists for $k$ sufficiently large), then notice that (for convenience we denote $T$ as $\mathcal{A}_0$)
\begin{align*}
    \sum_{i=1}^{p}(\mathcal{A}_{\omega_p^i})^{\otimes k} &= \sum_{i=1}^{p}(\mathcal{T}+\omega_p^i\mathcal{A}_1+\ldots+\omega_p^{id}\mathcal{A}_d)^{\otimes k},
\end{align*} where $\omega_p$ is a primitive $p$-th root of unity. On the right-hand-side, the coefficient for $\bigotimes_{j=1}^{k} \mathcal{A}_{\alpha_j}$ for all $0 \leq \alpha_1,\ldots,\alpha_k \leq d$ is  
\[
\sum_{i=1}^{p}\omega_p^{i\cdot \sum_{j=1}^{k}\alpha_j},
\] which equals $0$ unless $\alpha_j = 0$ for all $1 \leq j \leq k$. Therefore, we have 
\[
\mathcal{U} \leq \mathcal{T}^{\otimes k} = \sum_{i=1}^{p}\frac{1}{p}\cdot (\mathcal{A}_{\omega_p^i})^{\otimes k}
\] for $p \leq 2dk \leq O(k)$ many tensors. In addition, the rank of $\mathcal{A}_{\omega_p^i}$ is at most $r$ by our construction, which concludes our proof.
\end{proof}

This construction only works for infinite fields (with a primitive $p$-th root of unity). It is unclear whether a similar construction exists for finite fields $\mathbb{F}$. However, the construction would work over $\mathbb{F}[\omega_p]$, so we can circumvent this issue by simulating arithmetic operations on $\mathbb{F}[\omega_p]$ by arithmetic operations on $\mathbb{F}$.

\begin{lemma}
\label{lem: from border rank to rank finite fields}
Suppose there exists a matrix multiplication algorithm over $\mathbb{F}[\omega_p]$ with $T$ operations, then there exists a matrix multiplication algorithm over $\mathbb{F}$ with $T\cdot p^2$ operations.
\end{lemma}
\begin{proof}
We can simulate any arithmetic operation over $\mathbb{F}[\omega_p]$ by $O(p^2)$ arithmetic operations over $\mathbb{F}$ by considering each element of $\mathbb{F}[\omega_p]$ as a vector with length $O(p)$. The summing two elements in $\mathbb{F}[\omega_p]$ takes $O(p)$ operations, and multiplying them takes $O(p^2)$ operations.
\end{proof}

Therefore, we can always use the construction of \Cref{lem: from border rank to rank} with a possibly $O(p^2) \leq \poly(k)$ overhead.

\subsection{Proof of \Cref{thm: rectangular matrix multiplication}}
\label{sec: proof of main theorem}

We now put everything together and prove \Cref{thm: rectangular matrix multiplication}, and hence conclude the proof of \Cref{thm: main result}.

\begin{proof}[Proof of \Cref{thm: rectangular matrix multiplication}]

We will first prove that when $k = O((\log n)^{\varepsilon})$ for some $0<\varepsilon<0.01$,
\[
T(n^2,t,n^{2(k-1)}) \leq n^{O(1/(\log n)^{\frac{1}{3}-\varepsilon})}\cdot n^{2(k-1)}\cdot t.
\] This suffices because when $k$ is larger, we can find any $k' \leq O((\log n)^{\varepsilon})$ and apply Lemma \ref{lem:extendMM} such that
\begin{align*}
    T(n^2,t,n^{2(k-1)}) &\leq \frac{n^{2(k-1)}}{n^{2(k'-1)}}\cdot T(n^2,t,n^{2(k'-1)})\\
    &\leq \frac{n^{2(k-1)}}{n^{2(k'-1)}}\cdot n^{O(1/(\log n)^{\frac{1}{3}-\varepsilon})}\cdot n^{2(k'-1)}\cdot t\\
    &= n^{O(1/(\log n)^{\frac{1}{3}-\varepsilon})}\cdot n^{2(k-1)}\cdot t.
\end{align*}

Let $a = 1, b = \frac{1}{2}\log_n t<2, c = k-1$, and $P$ be a sufficiently large integer. Furthermore, let $q = (b+c)(1-\delta/2)$, $\alpha = \frac{P}{q+2}$, $L_1 = L_3 = \frac{\delta\alpha}{4}$, $L_2 = \alpha$, $N = (1-\delta/2)\alpha$. This is exactly the setting of Lemma \ref{lem: Josh flight}, by which we know 
\[
{(b+c)N \choose bN,cN}\cdot {aN+L_1+L_3 \choose aN,L_1,L_3}
\] is greater than both
\[
{(a+c)N \choose aN,cN}\cdot {bN+L_2+L_3 \choose bN,L_2,L_3},{(a+b)N \choose aN, bN}\cdot {cN+L_1+L_2\choose cN,L_1,L_2},
\] and that
\[
P = (a+b+c)N+L_1+L_2+L_3.
\] Therefore, we can apply Lemma \ref{lem: simpleCW} to obtain a tensor $H\odot \langle q^{aN},q^{bN},q^{cN}\rangle \leq \CW_q^{\otimes P}$ such that it has rank at most $O((q+2)^{P}\cdot P^2)$ and
\begin{align*}
    H \geq \frac{1}{2^{O(\sqrt{P})}}\cdot \frac{{P \choose L_1,L_2,L_3,aN,bN,cN}}{2{(b+c)N \choose bN,cN} \cdot{aN+L_1+L_3 \choose aN,L_1,L_3}+1} \geq \frac{{P \choose q\alpha,\alpha,\alpha}}{2^{O(\sqrt{P})}}.
\end{align*} (The last inequality comes from the proof of Lemma \ref{lem: Josh flight}.) Using Lemma \ref{lem: from border rank to rank} and Lemma \ref{lem: from border rank to rank finite fields}, over any fields, we can construct tensors $\mathcal{T}_1,\ldots,\mathcal{T}_{\ell}$ such that they all have rank at most $\ubar{R}(\CW_q) = q+2, \ell \leq P^2$ and 
\[
H\odot \langle q^{aN},q^{bN},q^{cN}\rangle \leq \CW_q^{\otimes P} \leq \sum_{i=1}^{\ell}\mathcal{T}_i^{\otimes P}.
\] Suppose $B_j = (A_{j1},A_{j2},A_{j3})$ is a bilinear algorithm for computing $\mathcal{T}_{j}$. The variables in $\mathcal{T}_j$ are exactly the same as the variables in $\CW_q$ by construction, so $T(A_{j1}),T(A_{j2}),T(A_{j3}) \leq 3(q+1)$ for all $j$ because each $\mathcal{T}_i$ is a tensor that comes from plugging in a certain value of $\lambda$ of Equation \ref{eq: CW tensor} above. In addition, Lemma \ref{lem: kronecker power upper bound} gives us
\[
T(A_{j1}^{\otimes P}) \leq (3q+3)\cdot P(q+2)^{P-1},
\] which is also an upper bound for $T(A_{j2}^{\otimes P}),T(A_{j3}^{\otimes P})$. 

We now set
\[
q = O(k) = O((\log n)^{\varepsilon/2}), N = O((\log n)^{2/3})
\] for a sufficiently small constant $\varepsilon>0$. Also notice that we can define 
\[
u:= \sum_{j=1}^{\ell}R(\mathcal{T}_i^{\otimes P}) \leq \sum_{j=1}^{\ell}(q+2)^P = (q+2)^P\cdot \ell
\] to be the sum of the ranks of tensors that we construct. Now we are ready to apply Theorem \ref{thm: the algorithm}. If $\frac{u}{H} \leq 2q^{(b+c)N}$, then $R(H\odot \langle q^{aN},q^{bN},q^{cN}\rangle) \leq 2q^{(b+c)N}$, which implies $R(\langle q^{aN},q^{bN},q^{cN}\rangle) \leq 2q^{(b+c)N}$. We can use this tensor directly to give an upper bound of $T(n^2,t,n^{2(k-1)})$:
\begin{align*}
    T(n^2,t,n^{2(k-1)}) &\leq 3q^{(b+c)N}\cdot n^{\frac{2}{N\log q}}\cdot  t\cdot n^{2(k-1)}.
\end{align*} To see that 
\[
3q^{(b+c)N}\cdot n^{\frac{2}{N\log q}} \leq n^{O(1/(\log n)^{\frac{1}{3}-\varepsilon})},
\] we can take $\log$ on both sides to obtain
\[
(b+c)N\cdot \log q +\frac{2\log n}{O((\log n)^{2/3}\log\log n)} \leq O((\log n)^{2/3+\varepsilon}).
\] This inequality holds because $b+c \leq O(q) = O((\log n)^{\varepsilon/2})$ and $N = O((\log n)^{2/3})$.

Therefore we can let $r:= R(\langle q^{aN},q^{bN},q^{cN}\rangle)$ and apply Theorem \ref{thm: the algorithm} to upper bound $T(n^2,t,n^{2(k-1)})$ as the sum of following three terms.
\begin{enumerate}
    \item $r\cdot \log_r H\cdot (n^{2(a+b)}+n^{2(a+c)}+n^{2(b+c)})$,
    \item $\Big(\frac{u}{H}\Big)^{\frac{2\log n}{N\log q}-\log_r H}\cdot H$,
    \item $\Big(\frac{u}{H}\Big)^{\frac{2\log n}{N\log q}-\log_r H-1}\cdot \sum_{j=1}^{\ell}\Big(3P(q+2)^{P+1}+(3q+3)\cdot P(q+2)^{P-1}\Big)$.
\end{enumerate} We prove that all three terms above are at most $n^{O(1/(\log n)^{\frac{1}{3}-\varepsilon})}\cdot n^{2(k-1)}\cdot t$.
Observe that 
\[
q^{(b+c)N}\cdot {P \choose q\alpha,\alpha,\alpha} = q^{q\alpha}{P \choose q\alpha,\alpha,\alpha} \geq \frac{(q+2)^P}{P^2}
\] (we prove this later in Lemma \ref{lem: p choose qalpha}). Since $u \leq O((q+2)^P\cdot P^2)$ and $H \geq {P \choose q\alpha,\alpha,\alpha}/2^{O(\sqrt{P})}$, we obtain 
\[
H \geq O\Big(\frac{u}{P^4\cdot q^{(b+c)N}\cdot 2^{O(\sqrt{P})}}\Big) = O\Big(\frac{u}{ q^{(b+c)N}\cdot 2^{O(\sqrt{P})}}\Big) \Longleftrightarrow \frac{u}{H} \leq q^{(b+c)N}2^{O(\sqrt{P})}.
\] In addition, we know from Lemma \ref{lem: Josh flight} that 
\[
H \leq {P \choose 2\alpha}\cdot {2\alpha \choose \alpha} \leq \Big(\frac{eP}{2\alpha}\Big)^{2\alpha}\cdot 2^{2\alpha} \leq O((2q)^{2\alpha}) = O((2q)^{\frac{2P}{q}}) = 2^{O(P\log q/q)}.
\] Now we upper bound the three terms.
\begin{enumerate}
    \item Since $a = 1, b = \frac{\log_n t}{2}, c = k-1, n^{2(b+c)} = n^{2(k-1)}\cdot t$. Therefore, it suffices to show that 
    \[
    r\cdot \log_r H \leq n^{O(1/(\log n)^{\frac{1}{3}-\varepsilon})}.
    \] Taking $\log$ on both sides, we get
    \[
    \log r+\log\log_r H \leq O((\log n)^{\frac{2}{3}+\varepsilon}).
    \] Since $r \leq q^{(a+b+c)N}$, $\log r \leq O(kN)\cdot \log q \leq O((\log n)^{\frac{2}{3}+\varepsilon})$. In addition, 
    \[
    \log\log_r H \leq O(\log P+\log\log q) \leq O(\log qN+\log\log q) \leq O((\log n)^{\frac{2}{3}+\varepsilon}).
    \]
    \item Since $\frac{u}{H} \leq q^{(b+c)N}2^{O(\sqrt{P})}$, after taking $\log$ it suffices to show that 
    \[
    \Big(\frac{2\log n}{N\log q}-\log_r H\Big)((b+c)N\log q+O(\sqrt{P}))+\log H \leq O((\log n)^{2/3+\varepsilon})+2(k-1)\log n+\log t.
    \] The left-hand-side is upper bounded by
    \begin{align*}
        &2(b+c)\log n+O\Big(\frac{2\log n\sqrt{qN}}{N\log q}\Big)+\log H\\
        &\leq 2(k-1)\log n+\log t+O\Big(\frac{\log n}{\sqrt{N}}\Big)+O\Big(\frac{P\log q}{q}\Big)\\
        &\leq 2(k-1)\log n+\log t+O\Big(\frac{\log n}{\sqrt{N}}\Big)+O\Big(N\log q\Big).
    \end{align*} Since $N = O((\log n)^{2/3})$, we know $O\Big(\frac{\log n}{\sqrt{N}}\Big) \leq O((\log n)^{2/3})$ and $O(N\log q) \leq O((\log n)^{2/3+\varepsilon})$.
    \item Since $\ell \leq P^2$, the third term can be upper bounded by 
    \[
    \Big(\frac{u}{H}\Big)^{\frac{2\log n}{N\log q}-\log_r H-1}\cdot P^2\cdot (q+2)^{P+1}.
    \] To see that this is smaller than $n^{O(1/(\log n)^{\frac{1}{3}-\varepsilon})}\cdot n^{2(k-1)}\cdot t$, we again take $\log$ on both sides. The third term will become 
    \begin{align*}
        &\Big(\frac{2\log n}{N\log q}-\log_r H-1\Big)\cdot \Big((b+c)N\log q+O(\sqrt{P})\Big)+O(\log P)+(P+1)\log (q+2) \\
        &\leq 2(k-1)\log n + \log t+O(\log\log n+P\log q)\\
        &= 2(k-1)\log n + \log t+O(\log\log n+qN\log q)\\
        &\leq 2(k-1)\log n+ \log t +O((\log n)^{2/3+\varepsilon}).
    \end{align*}
\end{enumerate} Therefore, we have shown that 
\[
T(n^2,t,n^{2(k-1)}) \leq n^{O(1/(\log n)^{\frac{1}{3}-\varepsilon})}\cdot n^{2(k-1)}\cdot t.
\] We still need to prove that 
\[
T(t,n^2,n^{2(k-1)}) \leq n^{O(1/(\log n)^{\frac{1}{3}-\varepsilon})}\cdot n^{2(k-1)}\cdot t.
\] However, this can be done by a very similar proof as above where we switch $a,b$ throughout the proof. We used the bilinear algorithms $B_j = (A_{j1},A_{j2},A_{j3})$ to compute 
\[
H\odot \langle q^{aN},q^{bN},q^{cN}\rangle \leq \sum_{i=1}^{\ell}\mathcal{T}_i^{\otimes P}
\] for $a = 1, b = \frac{\log_n t}{2}, c = k-1$. When we switch $a$ and $b$, we can use $B_j' = (A_{j1},A_{j3}^{\top},A_{j2}^{\top})$ to compute $H\odot \langle q^{bN},q^{aN},q^{cN}\rangle$. The cost associated with this new bilinear algorithm is the same using Lemma \ref{lem: matrix inverse time}.
\end{proof}

\begin{lemma}
\label{lem: p choose qalpha}
Suppose $P = (q+2)\alpha$ for some positive integers $q,\alpha,P$, then
\[
q^{q\alpha}{P \choose q\alpha,\alpha,\alpha} \geq \frac{(q+2)^P}{P^2}.
\]
\end{lemma}
\begin{proof}
We can lower bound ${P \choose q\alpha,\alpha,\alpha}$ by binary entropy function as 
\begin{align*}
    {P \choose q\alpha,\alpha,\alpha} = {P \choose q\alpha,2\alpha}\cdot {2\alpha \choose \alpha,\alpha}
    \geq \frac{1}{P+1}\cdot 2^{P\cdot H(2/(q+2))}\cdot \frac{1}{2\alpha}\cdot 2^{2\alpha \cdot H(1/2)}.
\end{align*} Since $(P+1)2\alpha \leq P^2$, it suffices to show that 
\[
q^{q\alpha}\cdot 2^{P\cdot H(2/(q+2))+2\alpha\cdot H(1/2)} \geq (q+2)^P.
\] Taking $\log$ on left-hand-side gives us
\begin{align*}
    & P\Big(\frac{2}{q+2}\log \frac{q+2}{2}+\frac{q}{q+2}\log \frac{q+2}{q}\Big)+2\alpha\log 2+q\alpha\log q\\
    &= P\Big(\log(q+2)-\frac{2}{q+2}\log 2-\frac{q}{q+2}\log q\Big)+2\alpha \log 2+q\alpha \log q\\
    &\geq P\log (q+2)
\end{align*} since $\alpha,q \geq 1$.
\end{proof}

\section{Applications}

In this section we show that our algorithms can be used to improve the leading constant of:
\begin{enumerate}
    \item the current fastest matrix multiplication algorithm via the Laser method,
    \item matrix multiplication algorithms designed via the group-theoretic approach, and
    \item matrix multiplication algorithms assuming $\omega = 2$.
\end{enumerate}
We will mainly apply Theorem \ref{thm: the algorithm} to different scenarios to obtain improvements and compare the leading constant that we obtain and the leading constant obtained via the normal recursive approach.

\subsection{Improving the Current Best Algorithm}
\label{sec: application to current best algo}

In this section we will prove Theorem \ref{thm:maincw} and show that our approach can be applied to the current best fast matrix multiplication algorithm to improve its leading constant. We first state the current best result and briefly explain its techniques.

\begin{theorem}[\cite{DWZ22,williams2023new,ADWXXZ24}] \label{thm:curromega}
Let $\omega_0 < 2.372$ denote the current best bound on $\omega$. Then, there is a function $f : \mathbb{N} \to (0,1)$ with $f(N) = \frac{1}{\Theta(\sqrt{\log N})}$ such that: For all positive integers $N$, there exist integers $m,H,C$ with $m,H = 2^{\Theta(N)}$ and $C \leq 2^{O(N / \log N)}$ such that
\[
H\odot \langle m,m,m\rangle \leq C \odot \CW_5^{\otimes N}
\] and 
\[
\frac{C \cdot R(\CW_5^{\otimes N})}{H} \leq m^{\omega_0+f(N)}.
\]
\end{theorem}

\begin{proof}[Proof Sketch]
    This is the result of the laser method application of \cite{williams2023new,ADWXXZ24}, which ultimately applies the asymptotic sum inequality to conclude that $$\omega < \frac{\log \left( \frac{C \cdot R(\CW_5^{\otimes N})}{H} \right)}{\log m},$$ and then takes the limit as $N \to \infty$. The parameter $C$ denotes the number of copies of the tensor one must make to fix the `holes' which appear in the matrix multiplication tensors in their construction; see \cite[Corollary 4.2; see also Sections 5.6, 6.6 and 7]{williams2023new}.
\end{proof}

Before proceeding, we calculate what leading coefficient we would obtain if we are to design an algorithm with this result using the standard recursive approach.

\begin{remark}
\label{rmk: remark}
Suppose we are aiming for an algorithm with exponent $\omega_0 + \varepsilon$. We thus require $f(N) = \frac{1}{\Theta(\sqrt{\log N})}<\varepsilon$, and so we need $\log N = \Theta(1/\varepsilon^2)$ and therefore
$N = 2^{\Theta(1/\varepsilon^2)}$. As a consequence, $m,H = 2^{\Theta(N)} = 2^{2^{\Theta(1/\varepsilon^2)}}$. When applying Theorem \ref{thm:curromega} for multiplying matrices, the current algorithm in \textup{\cite{DWZ22,williams2023new,ADWXXZ24}} uses the asymptotic sum inequality (Lemma \ref{lem: Schonhage thm}) which shows that for any $s \geq 1$,
\begin{align*}
    R(\langle m^s,m^s,m^s\rangle) &\leq \Big(\frac{C\cdot R(\CW_5^{\otimes N})}{H}\Big)^{s}\cdot H.
\end{align*} Then one can use the tensor $\langle m^s,m^s,m^s\rangle$ to recursively multiply $m^k \times m^k$ matrices with an operation count 
\begin{align*}
    T(m^k,m^k,m^k) &\leq 3m^{2s}\cdot (m^k)^{\frac{\log R(\langle m^s,m^s,m^s\rangle)}{\log m^s}}\\
    &= 3m^{2s}\cdot (m^k)^{\frac{\log (R(H\odot \langle m,m,m\rangle)/H)}{\log m}+\frac{\log H}{s\log m}}\\
    &= 3m^{2s}\cdot \Big(\frac{R(H\odot \langle m,m,m\rangle)}{H}\Big)^k\cdot H^{k/s}\\
    &\leq 3m^{2s}\cdot H^{k/s}\cdot (m^{\omega_0+f(N)})^{k}.
\end{align*} Therefore, the leading coefficient $3m^{2s}\cdot H^{k/s}$ of the usual algorithm can be optimized to 
\[
H^{O\big(\sqrt{\frac{k\log m}{\log H}}\big)} = 2^{O(N\sqrt{k})}\leq 2^{2^{O(1/\varepsilon^2)}\cdot \sqrt{k}}
\] for some $s = \Theta\Big(\sqrt{\frac{k\log H}{\log m}}\Big)$. In that case, we have
\[
T(m^k,m^k,m^k) \leq m^{2s}\cdot H^{k/s}\cdot (m^k)^{\omega_0+f(N)}\leq 2^{2^{O(1/\varepsilon^2)}\cdot \sqrt{k}}\cdot (m^k)^{\omega_0+\varepsilon}. 
\] Let $M = m^k$ such that we are multiplying $M \times M$ matrices. The leading coefficient scales exponentially with $\sqrt{k} = \sqrt{\frac{\log M}{\log m}}$ with a huge constant depending on $\varepsilon$. We use $C_0(\varepsilon)$ to represent this huge constant $2^{2^{O(1/\varepsilon^2)}}$ such that the leading coefficient is 
\[
C_0(\varepsilon)^{\sqrt{k}} = C_0(\varepsilon)^{\sqrt{\frac{\log M}{\log m}}}.
\] However, in order for this approach using $\langle m,m,m \rangle$ to be valid, we must have $M \geq m$ since $k\geq 1$, which means 
\[
\frac{\log M}{\log m} \geq \frac{\max\{\log M,\log m\}}{\log m} \geq \frac{\log M+\log m}{2\log m} = \frac{1}{2}+\frac{\log M}{2\log m}.
\] The leading coefficient is therefore at least 
\[
C_0(\varepsilon)^{\sqrt{\frac{\log M}{\log m}}} \geq C_0(\varepsilon)^{\sqrt{\frac{1}{2}+\frac{\log M}{2\log m}}} \geq C_0(\varepsilon)^{\frac{1}{2}+\frac{1}{2}\sqrt{\frac{\log M}{2\log m}}}.
\] Finally, define $C(\varepsilon) = \sqrt{C_0(\varepsilon)}$ such that 
\[
C_0(\varepsilon)^{\frac{1}{2}+\frac{1}{2}\sqrt{\frac{\log M}{2\log m}}} = C(\varepsilon)^{1+\sqrt{\frac{\log M}{2\log m}}}.
\] This is the $C(\varepsilon)$ which appears in Theorem~\ref{thm:maincw} above.
\end{remark}

We now show that we could give a better upper bound that is independent with respect to the input size using our algorithm. 

\begin{theorem}[Theorem \ref{thm:maincw}] \label{thm:improvedcurromega}
Let $\omega_0<2.372$ denote the current best bound on $\omega$. Then there is a function $f:\mathbb{N} \rightarrow (0,1)$ with $f(N) = \frac{1}{\Theta(\sqrt{\log N})}$ such that: For all positive integers $N$, there exists integers $m,H,C$ with $m,H = 2^{\Theta(N)}$ and $C \leq 2^{O(N/\log N)}$ such that
\begin{align*}
    T(m^k,m^k,m^k) &\leq H^{1/\Theta(\sqrt{\log N})}\cdot m^{(\omega_0+f(N))k}+O(m^{\omega_0+f(\sqrt{N})}\cdot N)\cdot m^{2k}\\
    &\leq 2^{O(N/\sqrt{\log N})}\cdot m^{(\omega_0+f(N))k}+2^{O(N\log N)}\cdot m^{2k}
\end{align*} for sufficiently large $k$.
\end{theorem}

Before proving Theorem \ref{thm:improvedcurromega}, we first compare our new leading constant with the previous leading coefficient. Our new leading constant is at most 
\[
2^{O(N/\log N)} \leq 2^{\frac{2^{O(1/\varepsilon^2)}}{1/\varepsilon^2}} \leq C(\varepsilon)^{O(\varepsilon^2)},
\] which not only does not depend on the input size but also is strictly better since $O(\varepsilon^2)<\frac{1}{2}$. In addition, our coefficient for the lower order term $m^{2k}$ also does not depend on the input size and is not too big compared to the leading constant:
\[
2^{O(N\log N)} \leq 2^{2^{O(1/\varepsilon^2)}/\varepsilon^2} \leq C(\varepsilon)^{O(1/\varepsilon^2)}.
\]

\begin{proof}
    We use Theorem \ref{thm:curromega} as the tensor needed by our algorithm in Theorem \ref{thm: the algorithm} (with $a=b=c=1$). For parameter $N$ that we will select, let $H,m,C,f$ be as in Theorem \ref{thm:curromega} and let $r = R(\langle m,m,m\rangle)$. 
    
    Theorem \ref{thm: the algorithm} gives us that for any sufficiently large positive integer $k$,
    \begin{align*}
        T(m^k,m^k,m^k) &\leq r\log_r H\cdot 3m^{2k} + H(m^{\omega_0+f(N)})^{k-\log_r H} + (m^{\omega_0+f(N)})^{k-\log_r H}\cdot CN \cdot O(N \cdot 7^N)\\
        &=3r\log_r H\cdot m^{2k}+\frac{H}{m^{(\omega_0+f(N))\log_r H}}\cdot m^{(\omega_0+f(N))k}+\frac{ O(CN^2\cdot 7^N)}{m^{(\omega_0+f(N))(\log_r H+1)}}\cdot m^{(\omega_0+f(N))k}.
    \end{align*} By Lemma \ref{lem: from border rank to rank} we know $O(7^N) \leq R(\CW_5^{\otimes N}) \leq O(N \cdot 7^N)$, and therefore $O(7^N) \leq \frac{H\cdot m^{\omega_0+f(N)}}{C}$. We substitute this to get:
    \[
    T(m^k,m^k,m^k) \leq 3r\log_r H\cdot m^{2k}+\frac{H}{m^{(\omega_0+f(N))\log_r H}}\cdot m^{(\omega_0+f(N))k}+\frac{O(N^2)}{m^{(\omega_0+f(N))\cdot\log_r H}}\cdot m^{(\omega_0+f(N))k}.
    \] By our choice of $m,H,N$, 
    \[
    \frac{O(N^2)}{m^{(\omega_0+f(N))\cdot \log_r H}} \leq o(1).
    \] Let $C_1,C_2,C_3>0$ be constants such that 
    \[
    H \leq 2^{C_1N}, m \leq 2^{C_2N}, C \leq 2^{C_3N/\log N}.
    \] We use the result in Theorem \ref{thm:curromega} with $\sqrt N$ to obtain
        \[
        2^{C_1\sqrt{N}}\odot \langle 2^{C_2\sqrt{N}},2^{C_2\sqrt{N}},2^{C_2\sqrt{N}}\rangle \leq 2^{C_3\sqrt{N}/\log \sqrt{N}}\cdot \CW_5^{\otimes\sqrt{N}}.
        \] Lemma \ref{lem: Schonhage thm} (with $s = \sqrt{N}$) implies that 
        \begin{align*}
            R(2^{C_1\sqrt{N}}\odot \langle 2^{C_2\sqrt{N}},2^{C_2\sqrt{N}},2^{C_2\sqrt{N}}\rangle^{\otimes \sqrt{N}}) &\leq \Big(\frac{2^{C_3\sqrt{N}/\log \sqrt{N}}\cdot R(\CW_5^{\otimes \sqrt{N}})}{2^{C_1\sqrt{N}}}\Big)^{\sqrt{N}}\cdot 2^{C_1\sqrt{N}}\\
            &\leq (2^{C_2\sqrt{N}})^{(\omega_0+f(\sqrt{N}))\sqrt{N}}\cdot 2^{C_1\sqrt{N}}\\
            &= (2^{C_2N})^{\omega_0+f(\sqrt{N})+\frac{C_1}{C_2\sqrt{N}}}\\
            &= m^{\omega_0+f(\sqrt{N})+\frac{C_1}{C_2\sqrt{N}}}.
        \end{align*} Therefore, we know 
        \[
        r \leq R(\langle 2^{C_2N},2^{C_2N},2^{C_2N}\rangle) \leq R(2^{C_1\sqrt{N}}\odot \langle 2^{C_2\sqrt{N}},2^{C_2\sqrt{N}},2^{C_2\sqrt{N}}\rangle^{\otimes \sqrt{N}}) \leq m^{\omega_0+f(\sqrt{N})+\frac{C_1}{C_2\sqrt{N}}},
        \] which implies that (recall $f(N) = \frac{1}{\Theta(\sqrt{\log N})}$)
        \begin{align*}
            (\omega_0+f(N))\log_r m &\geq \frac{\omega_0+f(N)}{\omega_0+f(\sqrt{N})+\frac{C_1}{C_2\sqrt{N}}}\\
            &= 1 -\frac{f(\sqrt{N})-f(N)+\frac{C_1}{C_2\sqrt{N}}}{\omega_0+f(\sqrt{N})+\frac{C_1}{C_2\sqrt{N}}}\\
            &\geq 1-\frac{1}{\Theta(\sqrt{\log N})}.
        \end{align*} Therefore, the leading constant of our algorithm is upper bounded by 
        \[
        H^{1/\Theta(\sqrt{\log N})} + o(1) \leq 2^{O(N/\sqrt{\log N})}
        \] and the coefficient for the lower-order term $m^{2k}$ is 
        \[
        3r\log_r H \leq O(m^{\omega_0+f(\sqrt{N})}\cdot N) \leq 2^{O(N\log N)}.
        \]
\end{proof}

\subsection{Applications to the Group-Theoretic Method}
\label{sec: application to the group method}

The group theoretic method for multiplying matrices was introduced by \cite{CU03} and refined by \cite{CKSU05}. The key idea behind this method is to use the tensor $\mathcal{T}_G$ of the group algebra of a finite group $G$ as the `intermediate' tensor in designing matrix multiplication algorithms, rather than $\CW_q$.

As in $\CW_q$-based algorithms, we aim to design algorithms using $\mathcal{T}_G$ by finding a zeroing out \begin{align}\label{groupeq2}\mathcal{T}_G \geq H \odot \langle m,m,m \rangle\end{align} for $m,H$ as large as possible. In the group-theoretic method, we generally aim to use combinatorial or algebraic properties of $G$ to find this zeroing out (by finding sets satisfying the so-called \emph{simultaneous triple product property}~\cite{CKSU05}).

One particularly appealing reason to use $\mathcal{T}_G$ is that one can use the representation theory of $G$ to prove that it has low rank via very structured encoding and decoding matrices. A key ingredient in these matrices is the Discrete Fourier Transform over $G$, a matrix $F_G \in \mathbb{C}^{|G| \times |G|}$ consisting of all the irreducible representations of all the elements of $G$. (Its exact definition does not matter much for our results here, but we note that a recent algorithm by Umans~\cite{umans2019fast} showed that $T(F_G) \leq |G|^{\omega/2 + o(1)}$.)

\begin{proposition} \label{prop:gfft}
    Let $G$ be a finite group, and let $\irr(G)$ be the set of irreducible representations of $G$. For a $\rho \in \irr(G)$, let $d_\rho$ denote its dimension.

    For each $\rho \in \irr(G)$, let $r_\rho$ be a positive integer and let $M_{\rho,1}, M_{\rho,2} \in \mathbb{F}^{r_\rho \times d_\rho^2}$ be encoding matrices, and $M_{\rho,3} \in \mathbb{F}^{D_\rho^2 \times r_\rho}$ be a decoding matrix, for the matrix multiplication tensor $\langle d_\rho, d_\rho, d_\rho \rangle$ (so, $r_\rho \geq R(\langle d_\rho, d_\rho, d_\rho \rangle)$).

    For $i \in \{1,2,3\}$, define the matrix $M_{G,i} := \bigoplus_{\rho \in \irr(G)} M_{\rho,i}$ to be the block-diagonal matrix.

    Then, there is a bilinear algorithm for $\mathcal{T}_G$ with encoding matrices $M_{G,1} \times F_G$ and $M_{G,2} \times F_G$, and decoding matrix $F_G^{\top} \times M_{G,3}$.

    In particular, we may bound $$R(\mathcal{T}_G) \leq \sum_{\rho \in \irr(G)} R(\langle d_\rho, d_\rho, d_\rho \rangle)$$ and $$T(\mathcal{T}_G) \leq 3 T(F_G) + \sum_{\rho \in \irr(G)} T(\langle d_\rho, d_\rho, d_\rho \rangle) \leq 3 \cdot |G|^{1.5} + \sum_{\rho \in \irr(G)} T(\langle d_\rho, d_\rho, d_\rho \rangle).$$
\end{proposition}

See, for instance, \cite[Section 1.3]{CU03} for more details behind Proposition~\ref{prop:gfft}. Combining Equation (\ref{groupeq2}) and Proposition~\ref{prop:gfft}, together with the asymptotic sum inequality (Lemma \ref{lem: asymptotic sum inequality}) yields an upper bound on the matrix multiplication exponent $\omega$ via
$$\sum_{\rho \in \irr(G)} d_\rho^{\omega} \geq H \cdot m^\omega.$$

Since $H\odot \langle m,m,m\rangle \leq \mathcal{T}_{G}$, the straightforward way to convert this into an algorithm for multiplying $m^k \times m^k$ matrices is to use Sch{\"o}nhage's asymptotic sum inequality, Lemma~\ref{lem: asymptotic sum inequality} (as in \cite[Theorem 5.5]{CKSU05}) to obtain, for any positive integer $s$,
\[
R(\langle m^s,m^s,m^s\rangle) \leq \Big(\frac{R(\mathcal{T}_G)}{H}\Big)^s\cdot H.
\] This will then give us
\[
T(m^k,m^k,m^k) \leq 3m^{2s}\cdot \Big(\frac{R(\mathcal{T}_G)}{H}\Big)^k\cdot H^{k/s}.
\] The leading coefficient $3m^{2s}\cdot H^{k/s}$ can then be optimized to around $3H^{\sqrt{\frac{2k\log m}{\log H}}}$. Again the leading coefficient grows when the input size grows. However, we can use Theorem \ref{thm: the algorithm} to get rid of this dependence and take advantage of the structure of $\mathcal{T}_G$ which the asymptotic sum inequality isn't able to use.

\begin{theorem}
\label{thm: group final result}
Suppose that $G$ is a finite group and that $T_{G} \geq H\odot \langle m,m,m\rangle$ for some positive integers $N,H,m$ and $r = R(\langle m,m,m\rangle)$. Then we have 
\[
T(m^k,m^k,m^k) \leq \frac{16|G|^{1.5}}{(R(\mathcal{T}_G)/H)^{\log_r H+1}}\cdot \Big(\frac{R(\mathcal{T}_G)}{H}\Big)^k+ 3r\log_r H\cdot m^{2k}
\] for sufficiently large $k$.
\end{theorem}
\begin{proof}
We can combine any set of bilinear algorithms for computing $\langle d_{\rho},d_{\rho},d_{\rho}\rangle$ to a bilinear algorithm for $\mathcal{T}_G$. Therefore, Theorem \ref{thm: the algorithm} gives an upper bound for $T(m^k,m^k,m^k)$ as 
\begin{align*}
    &3r\log_r H\cdot m^{2k}+H\Big(\frac{R(\mathcal{T}_{G})}{H}\Big)^{k-\log_r H}+\Big(\frac{R(\mathcal{T}_{G})}{H}\Big)^{k-\log_r H-1}\cdot 3\Big(\sum_{\rho \in \irr(G)}T(\langle d_{\rho},d_{\rho},d_{\rho}\rangle)+|G|^{1.5}\Big)+|\mathcal{T}_G|
\end{align*} for some sufficiently small constant $\varepsilon>0$. The leading constant will therefore become 
\begin{align*}
    &\frac{H}{(R(\mathcal{T}_G)/H)^{\log_r H}}+\frac{3\Big(\sum_{\rho \in \irr(G)}T(\langle d_{\rho},d_{\rho},d_{\rho}\rangle)+|G|^{1.5}\Big)+|G|}{(R(\mathcal{T}_G)/H)^{\log_r H+1}}.
\end{align*} Notice that 
\[
R(\mathcal{T}_{G}) \leq \sum_{\rho \in \irr(G)}R(\langle d_{\rho},d_{\rho},d_{\rho}\rangle) \leq \sum_{\rho \in \irr(G)}T(\langle d_{\rho},d_{\rho},d_{\rho}\rangle),
\] which implies that 
\begin{align*}
   \frac{\sum_{\rho \in \irr(G)}T(\langle d_{\rho},d_{\rho},d_{\rho}\rangle)}{(R(\mathcal{T}_G)/H)^{\log_r H+1}} 
   = \frac{H}{(R(\mathcal{T}_G)/H)^{\log_r H}}\cdot \frac{\sum_{\rho \in \irr(G)}T(\langle d_{\rho},d_{\rho},d_{\rho}\rangle)}{R(\mathcal{T}_G)}\geq \frac{H}{(R(\mathcal{T}_G)/H)^{\log_r H}}.
\end{align*} The leading constant is therefore upper bounded by
\begin{align*}
    \frac{4\Big(\sum_{\rho \in \irr(G)}T(\langle d_{\rho},d_{\rho},d_{\rho}\rangle)+|G|^{1.5}\Big)}{(R(\mathcal{T}_G)/H)^{\log_r H+1}}&\leq  \frac{12\sum_{\rho \in \irr(G)}d_{\rho}^3+4|G|^{1.5}}{(R(\mathcal{T}_G)/H)^{\log_r H+1}}\\
    &\leq \frac{12\max_{\rho \in \irr(G)}d_{\rho}\cdot \sum_{\rho \in \irr(G)}d_{\rho}^2+4|G|^{1.5}}{(R(\mathcal{T}_G)/H)^{\log_r H+1}}\\
    &= \frac{16|G|^{1.5}}{(R(\mathcal{T}_G)/H)^{\log_r H+1}},
\end{align*} where we use the fact that $\sum_{\rho \in \irr(G)}d_{\rho}^2 = |G|$.
\end{proof}

Similar to the previous section, the old ``leading constant'' scaled with the input size, but our improved leading constant does not. 
Directly comparing the constants depends on the values of $H,m$ achieved by the embedding of matrix multiplication into $\mathcal{T}_G$. For one example, many such algorithms have $H,m = |G|^{\Theta(1)}$, in which case the old leading constant is
\[
H^{\sqrt{\frac{24k\log m}{\log H}}} = H^{\Theta(\sqrt{k})} = |G|^{\Theta(\sqrt{k})},
\] whereas the new leading constant is the improved
\[
\frac{16|G|^{1.5}}{(R(\mathcal{T}_G)/H)^{\log_r H+1}} \leq 16 \cdot |G|^{1.5}.
\] This is an improvement even when $k$ (which scales with the logarithm of the input size) is a moderate constant.

We also note that the low-order terms $T(\langle d_{\rho},d_{\rho},d_{\rho}\rangle)$ in Theorem~\ref{thm: group final result} can be improved using our earlier results.
Suppose, for instance, we use tensor $\langle n_{\rho},n_{\rho},n_{\rho}\rangle$ with rank at most $t_{\rho}$ to compute $\langle d_{\rho},d_{\rho},d_{\rho}\rangle$, then the usual recursive method will give us
\[
T(\langle d_{\rho},d_{\rho},d_{\rho}\rangle) \leq 3n_{\rho}^2 \cdot d_{\rho}^{\log_{n_{\rho}}t_{\rho}},
\] but if $d_{\rho}$ is large enough compared to $n_{\rho}$, then Theorem \ref{thm: main result} allows us to improve this leading constant from $3n_{\rho}^2$ to $n_{\rho}^{O(1/(\log n_{\rho})^{0.33})}$. 


\subsection{The World Where $\omega = 2$}
\label{sec: application when omega is two}

It is popularly conjectured that the matrix multiplication exponent $\omega = 2$, which means that for all positive integers $n$, there is a tensor $\langle n,n,n\rangle$ with rank at most $n^{2+f(n)}$ for some function $f : \mathbb{N} \to (0,1]$ with 
\[
\lim_{n\rightarrow \infty}f(n) = 0.
\] The current best lower bound comes from~\cite{landsberg2014new,massarenti2013rank} which says $n^{2+f(n)} \geq 3n^2 - o(n^2)$, thus implying $f(n) \geq \frac{1}{O(\log n)}$. 

In this section, we ask what running time one gets for multiplying matrices form this assumption: What is $T(N,N,N)$? For what functions $f$ would we achieve $T(N,N,N) \leq \polylog N\cdot N^2$?

The straightforward recurrence one gets is $$T(N,N,N) \leq n^{2+f(n)} \cdot T(N/n,N/n,N/n) + O(n^2 N^2)$$
for any $n$ one chooses. This solves to $T(N,N,N) \leq O(n^2 N^{2 + f(n)})$. Even in the best-case scenario where $f(n) = c / \log n$ for some constant $c \geq 1$, this is minimized by picking $n = 2^{\sqrt{(c/2) \log N}}$, which yields $T(N,N,N) \leq 2^{O(\sqrt{\log N})} N^{2}$. In other words, using the straightforward recursive approach, it is impossible to achieve $T(N,N,N) \leq \polylog(N)\cdot N^2$, even if $f$ matches our current lower bounds.

Using our Theorem \ref{thm: main result}, one can get an improvement: for any $n$, we get an algorithm with operation count $T(N,N,N) \leq n^{1/O(\log\log n)} \cdot N^{2 + f(n)}$. To optimize this, we now pick $n = 2^{O(\sqrt{\log N \log \log N})}$, which yields $T(N,N,N) \leq 2^{O(\sqrt{\log N} / \log\log N)} N^{2}$. This small improvement is still unable to achieve $\polylog N\cdot N^2$.

In this section, we show that a further improvement is possible using a new recursive approach. Recall that when we designed our algorithm for Theorem \ref{thm: main result} above, we used the rectangular matrix multiplication algorithm of Lemma \ref{lem: rectangular matrix multiplication}. Our main new insight is that if $\omega=2$ and $f(n) < o(1/\log\log n)$, then this can directly be used to give a faster rectangular matrix multiplication algorithm than Lemma \ref{lem: rectangular matrix multiplication}. In other words, we use the fact that $\omega = 2$ not just to bound the exponent, but also in a recursive fashion to decrease the low-order terms of the algorithm.

\begin{lemma} \label{lem:fact}
Let $f : \mathbb{N} \to (0,1]$ be the minimum function such that $\langle n,n,n \rangle$ with rank at most $n^{2+f(n)}$ exists for all $n$.  Then, for any positive integers $n,k$, we have $f(n^k) \leq f(n)$.
\end{lemma}
\begin{proof}
    Taking the $k$th Kronecker power of $\langle n,n,n\rangle$ yields a construction $\langle n^k,n^k,n^k\rangle$ with rank at most $n^{2k+kf(n)}$, which gives an upper bound on $f(n^k)$.
\end{proof}

\begin{lemma} \label{lem:recurrence}
Let $f : \mathbb{N} \to (0,1]$ be the minimum function such that $\langle n,n,n \rangle$ with rank at most $n^{2+f(n)}$ exists for all $n$.  Define the function $g : \mathbb{N} \to \mathbb{R}$ by $g(n) = T(n,n,n) / n^{2+f(n)}$. Then, for any fixed $k \geq 3$, the function $g$ satisfies the recurrence
\[
g(n^k) \leq 9 n^{(k+2)f(n)} \cdot g(n^2).
\]
\end{lemma}
\begin{proof}
Define the function $g : \mathbb{N} \to \mathbb{R}$ by $g(n) = T(n,n,n) / n^{2+f(n)}$. Fix any $n \in \mathbb{N}$ and let $t = n^{2+f(n)}$. Our assumption that $\langle n,n,n\rangle$ with rank at most $t$ exists, combined with Lemma \ref{lem: rectangular matrix multiplication} and then two applications of Lemma \ref{lem:extendMM}, implies that for any $k \geq 3$ we have
\begin{equation*}
    \begin{split}
        T(n^k,n^k,n^k) &\leq t^k + \frac{4 t^{k-1}}{n^{2(k-1)}} \cdot T(t,n^2,n^{2(k-1)}) + \frac{2 t^{k-1}}{n^{2(k-1)}} \cdot T(n^2,t,n^{2(k-1)})\\
        &\leq t^k + \frac{4 t^{k-1}}{n^{2(k-1)}} \cdot \frac{t}{n^2} \cdot n^{2(k-2)} T(n^2,n^2,n^2) + \frac{2 t^{k-1}}{n^{2(k-1)}} \cdot n^{2(k-2)} \cdot T(n^2,t,n^2)\\
        &\leq t^k + \frac{4 t^{k-1}}{n^{2(k-1)}} \cdot \frac{t}{n^2} \cdot n^{2(k-2)} T(n^2,n^2,n^2) + \frac{2 t^{k-1}}{n^{2(k-1)}} \cdot n^{2(k-2)} \cdot \left(  \frac{t}{n^2} \cdot T(n^2,n^2,n^2) + tn^2 \right)\\
        &= 3t^k + \frac{6 t^{k-1}}{n^{2(k-1)}} \cdot \frac{t}{n^2} \cdot n^{2(k-2)} T(n^2,n^2,n^2)\\
        &= 3t^k + \frac{6 t^{k}}{n^{4}}  T(n^2,n^2,n^2)\\
        &= 3t^k + 6 t^{k} \cdot g(n^2) \cdot n^{2f(n^2)}\\
        &= 3 n^{2k + kf(n)}(1 + 2 \cdot g(n^2) \cdot n^{2f(n^2)}).
    \end{split}
\end{equation*}

Hence, using Lemma~\ref{lem:fact} to simplify, we get
\begin{equation*}
    \begin{split}
        g(n^k) &= T(n^k,n^k,n^k) / n^{2k + kf(n^k)} \\ 
        &\leq 3 n^{kf(n) - kf(n^k)}(1 + 2 \cdot g(n^2) \cdot n^{2f(n^2)})\\
        &\leq 3 n^{kf(n)}(1 + 2 \cdot g(n^2) \cdot n^{2f(n)})\\
        &\leq 9 n^{kf(n)} \cdot g(n^2) \cdot n^{2f(n)}\\
        &= 9 n^{(k+2)f(n)} \cdot g(n^2).
    \end{split}
\end{equation*}
\end{proof}

Solving the recurrence of Lemma \ref{lem:recurrence} can get somewhat messy, but we do so here in a number of important cases.

\begin{theorem}
    Suppose there is a constant $c$ such that $\langle n,n,n\rangle$ with rank at most $cn^2$ exists for all sufficiently large $n$, then $T(n,n,n) \leq n^2\cdot \polylog (n)$.
\end{theorem}

\begin{proof}
Let $f,g$ be as in Lemma~\ref{lem:recurrence}. For sufficiently large $n$, we have $f(n) \leq \log c / \log n$ so that $n^{f(n)} \leq c$. We therefore get the recurrence that, for sufficiently large $n$,
\[
g(n^k) \leq 9 n^{(k+2)f(n)} \cdot g(n^2) \leq 9 c^{k+2} g(n^2).
\]
Picking $k=4$ and substituting $m$ for $n^4$ gives that, for some positive constant $a$ and for sufficiently large $m$,
\[
g(m) \leq a \cdot g(\sqrt{m}).
\]
This solves to $g(m) \leq O(a^{\log \log m}) \leq \polylog m$, as desired.
\end{proof}

\begin{theorem}
    Suppose there is a constant $c$ such that $\langle n,n,n \rangle$ with rank at most $(\log n)^c n^2$ exists for all sufficiently large $n$, then $T(n,n,n) \leq n^2\cdot 2^{O((\log\log n)^2)}$.
\end{theorem}

\begin{proof}
For sufficiently large $n$, we have $f(n) \leq c \log \log n / \log n$ so that $n^{f(n)} \leq (\log n)^c$. We therefore get the recurrence
for sufficiently large $n$,
\[
g(n^k) \leq 9 n^{(k+2)f(n)} \cdot g(n^2) \leq (\log n)^{O(ck)} g(n^2).
\]
Again picking $k=4$ and substituting $m$ for $n^4$ gives that, for sufficiently large $m$,
\[
g(m) \leq (\log m)^{O(c)} \cdot g(\sqrt{m}).
\]
This solves to $g(m) \leq (\log n)^{O(\log \log m)} \leq 2^{O((\log\log m)^2)}$, as desired.
\end{proof}

We can bound the recurrence solution in general, although it may become fairly messy.

\begin{theorem}
For function $f : \mathbb{N} \to (0,1)$, define $$h_f(m) := \sum_{\ell=0}^{\log \log m} \frac{3}{2^{\ell+1}} \cdot f(m^{1/2^{\ell+2}}).$$ If $\langle n,n,n \rangle$ with rank at most $n^{2 + f(n)}$ exists for all $n$, then $T(n,n,n) \leq O(\log n) \cdot n^{2 + h_f(n)}$.
\end{theorem}

\begin{proof}
We start with the recurrence
\[
g(n^k) \leq 9 n^{(k+2)f(n)} \cdot g(n^2) .
\]
Again picking $k=4$ and substituting $m$ for $n^4$ gives that, for sufficiently large $m$,
\[
g(m) \leq O(m^{(3/2) f(\sqrt[4]{m})})\cdot g(\sqrt{m}).
\]
This solves to $$g(m) \leq O(\log m) \cdot m^{\sum_{\ell=0}^{\log \log m} {(3/2^{\ell+1}) f(m^{1/2^{\ell+2}})}}.$$



\end{proof}

\bibliographystyle{alpha}
\bibliography{sample}

\appendix

\section{The Leading Constant for the Current Best Matrix Multiplication Algorithm} \label{sec:backofenvelope}

In this section, we give an estimate of how large the integer $Q$ is such that the best-known matrix multiplication algorithms are achieved by bounding the rank of $\langle Q,Q,Q \rangle$. Even though the best known exponent is less than $2.4$, we focus on the modest question here of what $Q$ is needed to achieve exponent $2.5$, as the answer will already be very large. Rather than focus on the best implementations of the Coppersmith-Winograd approach, we instead focus here on the `simple' Coppersmith-Winograd construction of~\cite[Section~6]{CW87} which ultimately leads to a better $Q$ for this target of $2.5$. In a number of steps, we will use overly-optimistic bounds which we don't actually know how to achieve, which will only decrease the resulting value of $Q$.

Fix any positive integer $N$. Let $M = 2\binom{2N}{N} + 1$, and let $B \subseteq \mathbb{Z}_M$ be a set avoiding three-term arithmetic progressions, which may have size $$|B| \geq c_0  \frac{M \cdot \log^{1/4} M}{2^{2 \sqrt{2 \log_2 M}}}$$ for some constant $c_0$ by the construction of Elkin~\cite{elkin2010improved}.

Let $$A = \frac14 \frac{|B|}{M^2} \cdot \binom{3N}{N,N,N}.$$ \cite[Section~6]{CW87} shows that for any positive integer $N$, the tensor $$A \odot \langle 8^N,8^N,8^N \rangle$$ has rank at most $(1 + 12N) \cdot 1000^N$. (Here we are converting the border rank expression into a rank expression using~\cite{bini} and in particular the improved~\cite[Remark 6.5]{Blaser} for infinite fields.)

There are constants $c_1, c_2, c_3$ such that

$$\frac{c_1 2^{2N}}{\sqrt{N}} \leq \binom{2N}{N} \leq \frac{c_2 2^{2N}}{\sqrt{N}}$$

$$\binom{3N}{N,N,N} \geq \frac{c_3 3^{3N}}{N}$$

and thus for some constant $c_4$,

$$|B| \geq c_4 M \frac{N^{1/4}}{2^{4 \sqrt{N}}}.$$

Therefore, for some constant $c_5$,

$$A \geq c_5  \frac{N^{1/4}}{2^{4 \sqrt{N}}} \cdot \frac{\sqrt{N}}{2^{2N}} \cdot \frac{3^{3N}}{N} = c_5 \frac{3^{3N}}{N^{1/4} \cdot 2^{4\sqrt{N}} \cdot 2^{2N}}.$$

Applying Lemma~\ref{lem: Schonhage thm}, we see that for any positive integer $S$, the tensor $\langle 8^{SN}, 8^{SN}, 8^{SN}\rangle$ has rank at most
$$\left( \frac{(1 + 12N) 10^N}{A} \right)^S \cdot A \leq \frac{c_6^S N^{3S/4} 4000^{NS} \cdot 2^{4 \sqrt{N}S}}{3^{3NS}} \cdot \frac{3^{3N}}{N^{1/4} \cdot 2^{4\sqrt{N}} \cdot 2^{2N}} =: C$$
for some constant $c_6$. Since our goal is to get a lower bound on the leading constant of this algorithm, let us very optimistically assume $c_6 = 1$. The exponent of our algorithm is hence, optimistically,

$$\frac{\log(C)}{\log(8^{SN})}.$$

In the limit as $N=S \to \infty$, this gives the exponent $\log(4000/3^3) / \log(8) < 2.404$ as found by \cite[Section~6]{CW87}. However, how large do $N$ and $S$ actually need to be to get close to this exponent? We can use a calculator to compute, for instance, that:
\begin{itemize}
    \item If $N=S=10$ then we get exponent $2.956$.
    \item If $N=S=100$ then we get exponent $2.562$.
    \item If $N=S=250$ then we get exponent $2.500$.
    \item If $N=S=1000$ then we get exponent $2.450$.
\end{itemize}

Let us focus on just achieving exponent $2.5$. This approach hence needs $N=S=250$, and it achieves this by giving a bound on the rank of the tensor $\langle Q,Q,Q \rangle$ for $Q = 8^{NS} = 8^{250^2} \approx 1.33 \times 10^{56443}$. By comparison, the number of atoms in the visible universe is only about $10^{80}$. If we wanted to pick $N=S=10$ so that $8^{NS}$ is approximately $10^{80}$, then we would only get the exponent $2.956$, which is worse than the exponent of Strassen's algorithm! Indeed, the savings of this approach only kick in for very large $Q$.

\section{Generality of Sparse Decomposition} 
\label{sec: counting barrier}

In prior work, the sparse decomposition method was applied to speed up matrix multiplication algorithms using tensor $\langle n,n,n\rangle$ for small $n$. In this section, we investigate to what extent this method generalizes to improving the leading constants of larger identities. We focus here on the field $\mathbb{F}_2$, and on the case of applying an encoding matrix $X \in \mathbb{F}_2^{t \times n^2}$ for $\langle n,n,n\rangle$ (the case of a decoding matrix is similar).

The sparse decomposition method was introduced in \cite{KS20,BS19}. The main idea is that in order to upper bound $T(X^{\otimes k})$, one can first decompose $X = X_1\times X_2$ where $X_1$ is a $t \times r$ matrix and $X_2$ is a $r \times n^2$ matrix such that $X_1$ is sparse. Since $X^{\otimes k} = X_1^{\otimes k}\cdot X_2^{\otimes k}$, we have
\begin{align*}
    T(X^{\otimes k}) &\leq T(X_1^{\otimes k}) +T(X_2^{\otimes k})\\
    &\leq T(X_1)\cdot \frac{t^k-r^k}{t-r}+T(X_2)\cdot \frac{r^k-n^{2k}}{r-n^2}\\
    &\leq T(X_1)\cdot \frac{t^k-r^k}{t-r}+o(t^k),
\end{align*} where we use \Cref{lem: kronecker power upper bound}. The leading constant will thus become $\frac{T(X_1)}{t-r}+o(1)$. The main question then becomes: to what extent can we find such a factorization that improves the leading constant $\frac{T(X_1)}{t-r}$? We prove that $\Theta(n^2 / \log n)$ is tight: every matrix has a factorization which achieves this, and most matrices cannot possibly achieve better than this. 

To prove the lower bound, we start with a key counting lemma.

\begin{lemma} \label{lem:countsteps}
    For positive integers $t,r,c$, the number of matrices $X_1 \in \mathbb{F}_2^{t \times r}$ with $T(X_1) = c$ is at most $$(r+c)^t \cdot \prod_{i=1}^c \binom{r+i-1}{2}.$$
\end{lemma}

\begin{proof}
    An algorithm which gives an upper bound on $T(X_1)$ works as follows: Its input is a vector of $r$ values from $\mathbb{F}_2$, and at each step, the algorithm picks two current values and adds them. Thus, at step $i$, there are currently $r+i-1$ values that are in the algorithm's memory, and it can pick any two of them to add. Finally, in the last step, we must output a vector of length $t$, and there are $r+c$ values in memory that we could choose to output for each one.
\end{proof}

\begin{theorem}
    Suppose $n,t$ are positive integers with $n$ sufficiently large and $n^2 \leq t < n^{2.81}$. For most matrices $X \in \mathbb{F}_2^{t \times n^2}$, the minimum value of $\frac{T(X_1)}{t-r}$ over all $t > r \geq n^2$ and all factorizations $X = X_1 \times X_2$ with $X_1 \in \mathbb{F}_2^{t \times r}$ and $X_2 \in \mathbb{F}_2^{r \times n^2}$ is $\Omega(n^2 / \log n)$.
\end{theorem}

\begin{proof}
    The number of matrices $X_2 \in \mathbb{F}_2^{r \times n^2}$ is $2^{rn^2}$, and by Lemma~\ref{lem:countsteps}, the number of matrices $X_1 \in \mathbb{F}_2^{t \times r}$ with $T(X_1) \leq \frac15 n^2 (t-r) / \log n$ is at most \begin{align*}\left( r + \frac15 n^2 (t-r) / \log n \right)^t \cdot \prod_{i=1}^{\frac15 n^2 (t-r) / \log n} \binom{r+i-1}{2} &\leq n^{O(t)} \cdot \left( r + \frac{\frac15 n^2 (t-r)}{\log n} \right)^{\frac15 n^2 (t-r) / \log n} \\ &\leq \left( n^3 \right)^{\frac15 n^2 (t-r) / \log n + O(t)} \\ &\ll 2^{\frac34 n^2(t-r)}\end{align*}
    for sufficiently large $n$.

    Thus, the total number of matrices $X \in \mathbb{F}_2^{t \times n^2}$ with a factorization $X = X_1 \times X_2$ that achieves leading constant $\leq \frac14 n^2 (t-r) / \log n$ is at most
    $$\sum_{r=n^2}^{t-1} 2^{rn^2} \cdot 2^{\frac34 n^2(t-r)} = \sum_{r=n^2}^{t-1} 2^{\frac34 n^2t + \frac14 n^2 r} = \frac{2^{\frac34 n^2 t} \left(  2^{t(\frac14 n^2)} - 2^{n^2 (\frac14 n^2)} \right)}{2^{\frac14 n^2} - 1} \ll 2^{ n^2 t}.$$
    
    However, the number of matrices $X \in \mathbb{F}_2^{t \times n^2}$ is $2^{n^2 t}$, so most cannot have such a factorization, as desired.
\end{proof}

In fact, a similar counting argument shows that $O(n^2 / \log n)$ can always be achieved for any $X \in \mathbb{F}_2^{t \times n^2}$, so this is tight:

\begin{lemma}
\label{lem: exists_linearly_dependent_rows}
Suppose $\epsilon>0$ is a fixed constant, and $n,t$ are sufficiently large integers with $t > n^{2+\epsilon}$. Then, among any set $S \subseteq \mathbb{F}_2^{n^2}$ of size $|S| \geq t$, there is always a nonempty subset $S' \subseteq S$ of size $|S'| \leq O(n^2 / \log n)$ such that $\sum_{x \in S'} x = 0$.
\end{lemma}

\begin{proof}
    For a constant $a$, let $k = a n^2 / \log n$. The number of subsets $T \subseteq S$ of size $|T| = k$ is $\binom{t}{k}$, which is greater than $2^{n^2}$ for sufficiently large $a$. Thus, there must be two, $T, T'$, such that $\sum_{x \in T} x = \sum_{y \in T'} y$ by the pigeonhole principle. We can pick $S = T \oplus T'$ (the set of vectors in one but not both sets).
\end{proof}

\begin{lemma}
Given any $t \times n^2$ matrix $X$ over $\mathbb{F}_2$, there exists a decomposition $X = U\cdot \varphi$ such that $T(U) \leq O(n^2/\log n)$.
\end{lemma}
\begin{proof}
By \Cref{lem: exists_linearly_dependent_rows}, there exists $O(n^2/\log n)$ rows that are linearly dependent. Therefore, we can pick any one of them, say row $j$. Let $\varphi$ be the matrix after deleting row $j$ from $X$. Then we can pick $U$ to be a $t \times (t-1)$ matrix consisting of a $(t-1) \times (t-1)$ identity matrix (for the rows other than row $j$) and an extra row of at most $O(n^2/\log n)$ ones (for the linear combination which produces row $j$). We conclude that $T(U) \leq O(n^2/\log n)$.
\end{proof}

\end{document}